\documentclass[preprint,12pt,3p]{article}

\usepackage{amssymb}
\usepackage{amsthm}
\usepackage{amsmath}

\voffset = - 1.2cm
\theoremstyle{plain}
\newtheorem{theorem}{Theorem}
\newtheorem{corollary}[theorem]{Corollary}
\newtheorem{proposition}[theorem]{Proposition}

\theoremstyle{definition}
\newtheorem{definition}[theorem]{Definition}
\newtheorem{note}[theorem]{Note}

\numberwithin{equation}{section}
\numberwithin{theorem}{section}
\newcommand{\sbt}{\,\begin{picture}(-1,1)(-1,-3)\circle*{2}\end{picture}\ }

\begin{document}

\centerline{\Large \bf Dirac--Lie systems and Schwarzian equations}
\vskip 0.5cm
\centerline{ J.F. Cari\~nena$^{a)}$, J. Grabowski$^{b)}$, J. de Lucas$^{c)}$ and C. Sard\'on$^{d)}$}
\vskip 0.75cm

\centerline{$^{a)}$Faculty of Sciences and IUMA, University of Zaragoza,}
\vskip 0.2cm
\centerline{Pedro Cerbuna 12, 50.009, Zaragoza, Spain.}
\vskip 0.4cm
\centerline{$^{b)}$Institute of Mathematics, Polish Academy of Sciences,}
\vskip 0.2cm
\centerline{\'Sniadeckich 8,
P.O. Box 21, 00-956, Warszawa, Poland.}
\vskip 0.4cm
\centerline{$^{c)}$Department of Mathematical Methods in Physics, University of Warsaw,}
\vskip 0.2cm
\centerline{Ho\.za 74, 00-682, Warszawa, Poland.}
\vskip 0.4cm
\centerline{$^{d)}$Department of Fundamental Physics, University of Salamanca,}
\centerline{Plza. de la Merced s/n, 37.008, Salamanca, Spain.}

\begin{abstract}A {\it Lie system} is a system of differential equations admitting a {\it superposition rule}, i.e., a function describing its general solution in terms of any generic set of particular solutions and some constants.  Following ideas going back to the Dirac's description of constrained systems, we introduce and analyze a particular class of Lie systems on Dirac manifolds, called {\it Dirac--Lie systems}, which are associated with `Dirac--Lie Hamiltonians'. Our results enable us to investigate constants of the motion, superposition rules, and other general properties of such systems in a more effective way. Several concepts of the theory of Lie systems are adapted to this `Dirac setting' and new applications of Dirac geometry in differential equations are presented. As an application, we analyze solutions of several types of Schwarzian equations, but our methods can be applied also to other classes of differential equations important for Physics.
\end{abstract}

\noindent{\bf Keywords:}
Dirac structure;  Lie system; Poisson structure; Schwarzian equation; superposition rule; third-order Kummer--Schwarz equation; Vessiot--Guldberg Lie algebra

\vskip 0.2cm

\noindent{\bf MSC classification:} 34A26 (primary); 17B81 and 35Q53 (secondary)

\section{Introduction}

The study of Lie systems can be traced back to the end of the XIX century, when
K\"onigsberger \cite{Ko83},
Vessiot \cite{Ve93,Ve94}, and Guldberg \cite{Gu93} pioneered the analysis of
systems of first-order ordinary differential equations admitting a
{{\it superposition rule}} \cite{BM09II,LS}. In 1893 Lie succeeded in applying
his theory of Lie algebras \cite{Lie1880}
to characterize systems admitting a superposition rule \cite[Theorem 44]{LS}.
His result, known nowadays as {\it Lie--Scheffers theorem}
\cite{BM09II,CGM07},
states that a system of first-order ordinary differential equations admits a
superposition rule if and only if it describes
  the integral curves of a $t$-dependent vector field taking values in a
  finite-dimensional real Lie algebra of vector fields ({\it Vessiot--Guldberg Lie algebra}) \cite{Dissertationes}.

  During the XX century, Lie systems were almost forgotten
until the 80's, when Winternitz revived their study
\cite{AndHarWin82,PW}. He analyzed the classification of Lie
systems on $\mathbb{R}^2$ \cite{SW84,SW84b}, employed superposition rules
to study systems of first-order differential equations on supermanifolds
\cite{BecGagHusWin87,BecGagHusWin90}, and investigated Lie systems of
relevance  \cite{HWA83}.
His achievements boosted the study of Lie systems, which were found to have
many geometric properties and
applications in physics, mathematics, and control theory. For instance, the work \cite{Dissertationes}
details more than two hundred references on Lie systems and related topics.

It was recently noted that remarkable Lie systems admit a Vessiot--Guldberg Lie algebra
of Hamiltonian vector fields with respect to symplectic and Poisson structures  \cite{CGM00,CLS12Ham}, e.g.,  certain coupled Riccati equations \cite{BHLS12Cla}, Kummer--Schwarz and  second-order Riccati equations in Hamiltonian form \cite{CLS12Ham}, and others \cite{ADR11,Ru10}.
This gave rise to the definition of the so-called
{\it Lie--Hamilton systems} \cite{CLS12Ham} which enjoy a plethora
of geometric features
\cite{ADR11,BHLS12Cla,Dissertationes,CLS12Ham,Ru10,FLV10,BCHLS13,Gh13}.

Algebraic and Poisson geometric techniques have been employed to study Lie--Hamilton
systems \cite{BCHLS13,CLS12Ham,Gh13}. For instance, the superposition rule for Riccati equations can be
obtained in an algebraic way from a Casimir element of a certain Lie algebra \cite{BCHLS13}. Co-algebra techniques can also be applied to obtain superposition rules and constants of the motion
for these systems \cite{BCHLS13}. Additionally, other results concerning the integrability of these systems
have been found directly or indirectly from the geometric structures associated to Lie--Hamilton systems
\cite{ADR11,CLS12Ham,Ru10}.

Of course, not all Lie systems are
Lie--Hamilton systems. We here devise a simple and useful condition ensuring that a Lie system is not a Lie--Hamilton system, and we use it to show that Lie systems related to third-order Kummer--Schwarz equations \cite{LS12} and diffusion PDEs \cite{SSVG11} are not Lie--Hamilton systems.  Meanwhile, we prove that these systems admit a Vessiot--Guldberg Lie algebra of Hamiltonian vector fields, but this time with respect to a {\it presymplectic form} \cite{LM87}. The appearance of this new structure in several important Lie systems and the fact that such systems cannot be 	investigated by the methods of the theory of Lie--Hamilton systems motivate the present study.

We introduce a class of Lie systems covering all the aforementioned
Lie systems as particular instances: Dirac--Lie systems. Roughly speaking, a Dirac--Lie system is a Lie system possessing a Vessiot--Guldberg Lie algebra of Hamiltonian vector fields of a special and very general type. Most properties of standard Hamiltonian vector fields with respect to symplectic structures can naturally be extended to these `generalized' Hamiltonian vector fields. For example, these generalized Hamiltonian vector fields can be related to `generalized' Hamiltonian functions, the so-called admissible functions, which can be employed to study them and, as a byproduct, Dirac--Lie systems. In doing this, the standard techniques in Hamiltonian dynamics and Lie--Hamilton systems can be adapted to investigate Dirac--Lie systems, which are much more general than Lie--Hamilton ones.

More precisely, a {\it Dirac--Lie system} is a triple $(N,L,X)$ consisting of a Lie system $X$ on a manifold $N$ which
admits a Vessiot--Guldberg Lie algebra of Hamiltonian vector fields with respect to a
{\it Dirac structure} $L$ on $N$ \cite{Co87,Do87,JR11}. The latter is a maximally isotropic subbundle $L$ of the {\it Pontryagin bundle} $T^*N\oplus_N TN$ satisfying an integrability condition. Note that Dirac structures provide a geometric setting for Dirac's theory of constrained mechanical systems which generalizes simultaneously Poisson and presymplectic structures.

Using that Poisson and presymplectic manifolds can be described as particular cases of Dirac structures, we show that Lie--Hamilton systems, based upon Poisson manifolds, are a particular type of Dirac--Lie systems and we recover their properties as particular instances of our theory.

Dirac--Lie systems can be studied through
Dirac geometric techniques. This is more general and commonly easier than using techniques for
Lie--Hamilton systems. We prove that every Dirac--Lie system $(N,L,X)$ can be described by a $t$-dependent
Hamiltonian
$h:(t,x)\in \mathbb{R}\times N\mapsto h_t(x)\in \mathbb{R}$
whose functions $\{h_t\}_{t\in\mathbb{R}}$ span a
finite-dimensional real Lie algebra relative to the Poisson bracket of
admissible functions induced by $L$ \cite{Bu11,Co87}.
This suggests us to define a type of $t$-dependent Hamiltonians, called {\it Dirac--Lie
Hamiltonians}, that generalize the notion of {\it Lie--Hamiltonians} used for studying
Lie--Hamilton systems \cite{CLS12Ham}. Subsequently, we show that a Dirac--Lie system $(N,L,X)$ is
equivalent to a curve in a finite-dimensional real Lie algebra of sections of the Lie algebroid bracket induced on $L$ (see \cite{Ma01,Ma08} for an account of Lie algebroids).

We study {\it diagonal prolongations} of Dirac structures and Dirac--Lie systems which play a central r\^ole in determining superposition rules \cite{CGM07}. These notions are exploited to
analyze and to derive in an algebraic way $t$-independent constants of the motion, Lie symmetries, and superposition rules for Dirac--Lie systems. In order to illustrate our procedures, we obtain  a superposition rule for {\it Schwarzian  equations} \cite{SS12,TTX01,Ni97}, i.e. differential equations related to the {\it Schwarzian derivative} \cite{OT09} and also known as Schwarz equations \cite{Ts97}. Our method is
simpler than previous approaches based upon integrating systems of PDEs and/or ODEs \cite{LS12,PW}.

Further, we develop methods of generating new  Dirac--Lie systems out of an initial one. This results in the definition of the so-called {\it bi--Dirac--Lie systems}, i.e., Lie systems admitting a Vessiot--Guldberg Lie algebra of Hamiltonian vector fields with respect to two Dirac structures. This enables us to investigate $X$ through our previous results in two, generally non-equivalent, manners: by using $L$ and $L'$.
We devise a new procedure to produce bi--Dirac--Lie systems, based upon the use of $t$-independent Lie symmetries of $X$, that generalizes a previous result employed to study autonomous Hamiltonian
systems \cite{CMR02}. This is further generalized by using the so-called {\it gauge transformations} of Dirac structures \cite{BR03}.

All our previous techniques are applied to derive a mixed superposition rule for studying Schwarzian equations. The standard methods for deriving a
mixed superposition rule demands finding certain $t$-independent constants of the motion of a Lie system or integrating a system of ODEs \cite{GL12,LS12}. In both cases, it is necessary to integrate systems of PDEs/ODEs. In our case, since we aim at obtaining a mixed superposition rule for a Dirac--Lie system, the associated Dirac structure allows us to use purely algebraic-geometrical techniques to avoid integrating complicated systems of differential equations and to simplify the whole procedure.

We find out that our techniques can be applied to Schwarzian Korteweg-de Vries (SKdV) equations \cite{Mar11,I10}. This provides a new approach to the study of these equations. We derive soliton-type solutions for Schwarzian-KdV equations, namely shape-preserving traveling wave solutions. Moreover, we show how Lie systems and our methods can be applied to provide B\"acklund transformations for certain solutions of these equations. This can be considered as the first application of Dirac structures in studying PDEs of physical and mathematical interest from the point of view of the theory of Lie systems.

The structure of the paper is as follows. Sections 2, 3,  and 4 concern the notions used throughout our paper. In
Section 5, the analysis of several remarkable Lie systems that cannot be considered as Lie--Hamilton systems
leads us to introduce the concept of Dirac--Lie systems which encompasses
such systems as particular cases. Subsequently, the Dirac--Lie Hamiltonians
 are introduced and analyzed in Section 6. Next, we investigate
several geometric properties of Dirac--Lie systems in Section 7. Section 8 concerns the
study of constants of the motion and superposition rules of Dirac--Lie systems. Next, Section 9
is devoted to bi--Dirac--Lie systems. In Section 10, we illustrate the usefulness of all our methods to derive a mixed superposition rule \cite{GL12} to study Schwarzian equations. In Section 11 we devise an application of our techniques in SKdV equations. Finally, we summarize our main results and present an outlook of our future research  in Section 12.

\section{Dirac manifolds}
The concept of {\it Dirac structure}, proposed by Dorfman \cite{Do87} in the Hamiltonian framework of integrable evolution equations and defined in \cite{Co87} as a subbundle of the Whitney sum $T N\oplus_N T^\ast N$ (called the {\it extended tangent} or {\it Pontryagin bundle}) satisfying certain conditions, was thought-out as a common generalization of Poisson and presymplectic structures. It was designed also to deal with constrained systems, including constraints induced by degenerate Lagrangians, as
was investigated by Dirac \cite{Di50}, which is the reason for the name.				In this section, we present a brief survey on all the necessary notions and facts
(see for instance  \cite{Bu11,Co87,Co90,CW88,JR11,IV,NTZ12} for details).

We hereafter assume all mathematical objects to be real,
smooth, and globally defined. Manifolds are considered to be connected. This
permits us to omit several minor technical details while highlighting the main aspects of our theory. We hereafter call $\Gamma(E)$ the space of smooth sections of a bundle $(E,B,\pi:E\rightarrow B)$.

A {\it symplectic manifold} is a pair $\left( N, \omega \right) $,
where $N$ stands for a manifold and $\omega$ is a non-degenerate closed two-form
on $N$. We say that a
vector field $X$ on $N$ is Hamiltonian with respect to $(N,\omega)$ if there
exists a function $f\in C^\infty(N)$
such that
\begin{equation}\label{conHam}
\iota_X\omega=-df.
\end{equation}
In this case, we say that $f$ is a {\it Hamiltonian function} for $X$.
Conversely, given a
function $f$, there exists a unique vector field $X_f$ on $N$, the
so-called {\it Hamiltonian vector field} of $f$, satisfying (\ref{conHam}). This allows us to define a bracket $\{\cdot,\cdot \}:C^{\infty}
\left( N \right)\times C^\infty(N)\rightarrow C^\infty(N)$ given by
\begin{equation}\label{PB}
\left\{f,g\right\}=\omega(X_f,X_g)=X_f(g).
\end{equation}
This bracket turns $C^\infty(N)$ into a {\it Poisson algebra} $(C^\infty(N),\sbt\,,\{\cdot,\cdot\})$, i.e., $\{
\cdot, \cdot \}$ is a Lie bracket on $C^\infty(N)$ which additionally holds the
{\it Leibniz rule} with respect to
the standard product `${\bf \sbt}\,$' of functions:
\[ \{ f,g \sbt\, h\} = \{ f, g\} \sbt\, h + g\sbt\,\{ f, h\}\,, \qquad
   \forall f, g, h \in C^\infty(N). \]
For simplicity, we just hereafter write $fg$ for $f\sbt\,g$. The Leibniz rule can be rephrased by saying that $\{ f , \cdot\}$ is a derivation of the
associative algebra $(C^\infty(N),\sbt\,)$ for each $f \in C^\infty(N)$. Actually, this derivation is represented by the Hamiltonian vector field $X_f$.
The bracket
$\{\cdot,\cdot\}$ is called the
{\it Poisson bracket} of $(C^\infty(N),\sbt\,,\{\cdot,\cdot\})$. Note that if $(N,\omega)$ is a symplectic manifold, the non-degeneracy condition for
$\omega$ implies that $N$ is even dimensional \cite{FM}.

The above observations lead to the concept of a Poisson manifold which is a natural generalization of the symplectic one. A {{\it
Poisson manifold}} is a pair $\left( N,\{\cdot,\cdot\}\right)$, where
$\{\cdot,\cdot\}:C^\infty(N)\times
C^\infty(N)\rightarrow C^\infty(N)$ is the Poisson bracket of $(C^{\infty}( N),\sbt\,,\{\cdot,\cdot\})$ which is also referred to as a {\it Poisson structure} on $N$. In
view of this and (\ref{PB}), every symplectic manifold is a particular type of
Poisson manifold. Moreover, by noting that $\{f,\cdot\}$ is a derivation on
$(C^\infty(N),\sbt\,)$ for every $f\in C^\infty(N)$, we can associate with every function $f$  a
single vector field $X_f$, called the {\it Hamiltonian vector field} of $f$, such that $\{f,g\}=X_fg$ for all $g\in C^\infty(N)$, like in the symplectic case.

As the Poisson
structure is a derivation in each entry, it gives rise to a bivector field $\Lambda$, i.e., an element of $
\Gamma(\bigwedge^2TN)$, the referred to as {\it Poisson bivector}, such that
$\{f,g\}=\Lambda(df,dg)$. It is known that the Jacobi identity for
$\{\cdot,\cdot\}$ amounts to
$[\Lambda,\Lambda]_{SN}=0$, with $[\cdot,\cdot]_{SN}$ being the {\it
Schouten--Nihenjuis bracket} \cite{IV}. Conversely, a
bivector $\Lambda$ satisfying $[\Lambda,\Lambda]_{SN}=0$ gives rise to a Poisson bracket on $C^\infty(N)$ by
setting $\{f,g\}=\Lambda(df,dg)$. Hence, a Poisson manifold can be
considered, equivalently, as $(N,\{\cdot,\cdot\})$ or $(N,\Lambda)$. It is remarkable that $\Lambda$ induces
a bundle morphism $\widehat\Lambda:\alpha_x\in T^*N\rightarrow \widehat\Lambda(\alpha_x)\in TN$, where
$\bar\alpha_x(\widehat\Lambda(\alpha_x))=\Lambda_x(\alpha_x,\bar\alpha_x)$ for all $\bar\alpha_x\in T_x^*N$
, which enables us to write $X_f=\widehat \Lambda(df)$ for every $f\in C^\infty(N)$.

Another way of generalizing a symplectic structure is to consider a two-form $\omega$ which is merely closed (not necessarily of constant rank), forgetting the non-degeneracy assumption. In this case,
$\omega$ is said to be a {\it presymplectic form}
and the pair $\left( N, \omega \right)$ is called a
{\it presymplectic manifold} \cite{LM87}. Like in the symplectic case, we
call a vector field $X$ on $N$ {\it Hamiltonian} if there exists a function $f\in C^\infty(N)$, a {\it Hamiltonian function} for
$X$, such that (\ref{conHam}) holds for the presymplectic form $\omega$.

The possible degeneracy of $\omega$ introduces several differences
with respect to the symplectic setting.
For example, given an $f\in C^\infty(N)$,
we cannot ensure neither the existence nor the uniqueness of a vector field $X_f$
satisfying $\iota_{X_f}\omega=-df$. If it exists, we say that $f$ is an {\it
admissible function}
with respect to $\left( N, \omega \right)$. Since the linear combinations and multiplications of admissible functions are also admissible functions, the space ${\rm Adm}(N,\omega)$ of admissible functions of $(N,\omega)$ is a real associative algebra. It is canonically also a Poisson algebra. Indeed, observe that every
$f\in {\rm Adm}(N,\omega)$ is associated to a family of Hamiltonian vector fields of the form
$X_f+Z$, with $Z$ being a vector field taking values in $\ker \omega$. Hence,  (\ref{PB}) does not depend on the representatives $X_f$ and $X_g$ and becomes a Poisson bracket on the space ${\rm Adm}(N,\omega)$, making the latter into a Poisson algebra. It is also remarkable that
$$\iota_{[X_f,X_g]}\omega=\mathcal{L}_{X_f}\iota_{X_g}\omega-
\iota_{X_g}\mathcal{L}_{X_f}\omega=-\mathcal{L}_{X_f}dg=-d\{f,g\}\,.
$$
In consequence, $[X_f,X_g]$ is a Hamiltonian vector field with a Hamiltonian function $\{f,g\}$.

A natural question now arises: is there any geometric structure incorporating presymplectic and Poisson manifolds as particular cases? Courant \cite{Co87,Co90} provided an affirmative answer to this
question.

Recall that a {\it Pontryagin bundle} $\mathcal{P}N$ is a vector bundle $TN
\oplus_N T^{\ast} N$ on $N$.

\begin{definition} An {\it almost-Dirac manifold} is a pair $(N, L)$, where $L$ is a  maximally isotropic  subbundle of
$\mathcal{P}N$ with respect to the pairing
\[ \langle X_x + \alpha_x, \bar{X}_x + \bar{\alpha}_x\rangle _+ \equiv
\frac{1}{2}  (
   \bar{\alpha}_x(X_x) + \alpha_x(\bar{X}_x)),
\]
where $ X_x + \alpha_x, \bar{X}_x
+
   \bar{\alpha}_x \in T_xN \oplus
T_x^{\ast} N=\mathcal{P}_xN.$
In other words, $L$ is isotropic and has rank $n=\dim N$.

A {\it Dirac manifold} is an almost-Dirac manifold $(N,L)$  whose subbundle $L$,
its {\it Dirac structure}, is
involutive relative to the {\it Courant--Dorfman bracket} \cite{Co87,Do87,GG11,JR11}, namely
\[ [[X + \alpha, \bar{X} + \bar{\alpha}]]_C \equiv [X, \bar{X}] + \mathcal{L}_X
   \bar{\alpha} - \iota_{\bar{X}}d\alpha\,,\]
where
$X+\alpha, \bar X+\bar{\alpha}\in \Gamma(TN\oplus_NT^*N)$.
\end{definition}

Note that the Courant--Dorfman bracket satisfies the Jacobi identity in the form
\begin{equation}\label{Jacobi} [[\,[[e_1,e_2]]_C,e_3]]_C
\!=\![[e_1,[[e_2,e_3]]_C]]_C-\![[e_2,[[e_1,e_3]]_C]]_C,\,\, \forall e_1,e_2,e_3\!\in\! \Gamma(\mathcal{P}N),
\end{equation}
but is not skew-symmetric. It is, however, skew-symmetric on sections of the Dirac subbundle $L$, defining a {\it Lie algebroid} structure
$(L,[[\cdot,\cdot]]_C,\rho)$, where $\rho
: L\ni X_x + \alpha_x \mapsto X_x\in  TN$. This means that $(\Gamma(L),[[\cdot,\cdot]]_C)$ is a Lie algebra and the vector bundle morphism  $\rho:L\rightarrow TN$, the {\it anchor}, satisfies
\begin{equation}\label{anchor}[[e_1,fe_2]]_C=(\rho(e_1)f)e_2+f[[e_1,e_2]]_C
\end{equation}
for all $e_1,e_2\in \Gamma(L)$ and $f\in C^\infty(N)$ \cite{Co87}.
One can prove that, automatically, $\rho$ induces a Lie algebra morphism of $(\Gamma(L),[[\cdot,\cdot]]_C)$ into the Lie algebra of vector fields on $N$.
The generalized distribution $\rho(L)$, called the {\it characteristic distribution} of the Dirac structure, is therefore integrable in the sense of Stefan--Sussmann \cite{Su73}.

\begin{definition}
A vector field $X$ on $N$ is said to be an $L$-{\it Hamiltonian vector
field} (or simply a {\it Hamiltonian vector field} if $L$ is fixed) if there exists an
$f \in C^{\infty} (N)$ such that $X + df \in \Gamma(L)$. In this case, $f$
is an $L$-{\it Hamiltonian function} for $X$ and an {\it admissible function} of
$(N,L)$. Let us denote by
Ham$(N,L)$ and Adm$(N, L)$ the spaces of Hamiltonian vector fields and admissible functions of $(N,L)$, respectively.
\end{definition}

 The space ${\rm Adm}(N,L)$ becomes a Poisson algebra $( {\rm Adm} (N, L), \sbt\,, \{ \cdot, \cdot
\}_L)$ relative to the standard product of functions  and the Lie bracket
given by
\[ \{f, \bar f \}_L =X \bar f\,,\]
where $X$ is an $L$-Hamiltonian vector field for $f$.
Since $L$ is isotropic, $\{f,\bar f\}_L$ is
well defined, i.e., its value is independent on the choice of the $L$-Hamiltonian vector field associated to $f$. The elements $f\in {\rm Adm}(N,L)$
possessing trivial Hamiltonian vector fields are called the {\it Casimir functions} of
$(N,L)$ \cite{NTZ12}. We write ${\rm Cas}(N,L)$ for the set of Casimir functions of $(N,L)$. We can also distinguish the space $G(N,L)$ of $L$-Hamiltonian vector fields which admits
zero (or, equivalently, any constant) as an $L$-Hamiltonian function. We call them {\it gauge vector fields} of the Dirac structure.

Note that, if $X$ and $\bar X$ are $L$-Hamiltonian vector
fields with Hamiltonian functions $f$ and $\bar f$, then $\{f, \bar f \}_L$ is a Hamiltonian for $[X, \bar X]$:
\begin{equation*}
[[X+df,\bar X+d\bar f]]_C=[X,\bar X]+\mathcal{L}_{X}d\bar f-\iota_{\bar
X}d^2f=[X,\bar X]+d\{f,\bar f\}_L.
\end{equation*}
This implies that
$({\rm Ham}(N,L),[\cdot,\cdot])$ is a Lie algebra in which $G(N,L)$ is a Lie ideal.
Denote the quotient Lie algebra ${\rm Ham}(N,L)/G(N,L)$ by $\widehat{\rm Ham}(N,L)$.

\begin{proposition}\label{CasCon}
If  $(N, L)$ is a Dirac manifold, then  $\{{\rm Cas} (N, L), {\rm Adm}(N,L)\}_L=0$, i.e.,  {\rm Cas}$(N, L)$ is an ideal
  of the Lie algebra $( {\rm Adm} (N, L), \{ \cdot, \cdot \}_L)$. Moreover, we have the following exact
sequence of Lie algebra homomorphisms
\begin{equation}\label{PreSymSeqII}
0\hookrightarrow {\rm Cas}(N,L)\hookrightarrow {\rm Adm}(N,L)\stackrel{B_L}{\longrightarrow} \widehat{\rm Ham}(N,L)\rightarrow 0\,,
\end{equation}
with $B_L(f)=\pi(X_f)$, where the vector field $X_f$ is an $L$-Hamiltonian vector field of $f$, and $\pi$ is the canonical projection $\pi:{\rm Ham}
(N,L)\rightarrow \widehat{\rm Ham}(N,L)$.
\end{proposition}

For every Dirac manifold $(N,L)$, we have a canonical linear map
$\Omega^L_x : \rho (L)_x\subset T_xN \rightarrow \rho (L)_x^*\subset T^*_xN$ given by
\begin{equation}\label{twoform}
[\Omega^L_x (X_x)](\bar X_x) = -{\alpha}_x (\bar{X}_x), \qquad X_x,\bar X_x\in
\rho(L),
\end{equation}
where $\alpha_x\in T^*_xN$ is such that $X_x+\alpha_x\in L$.
Note that, as  $L$ is isotropic, $\Omega_x^L$ is well defined, i.e., the
value of
\[\Omega_x^L(X_x,\bar X_x)=[\Omega^L_x(X_x)](\bar X_x)\] is independent of the particular
${\alpha}_x$ and defines a skew-symmetric bilinear form $\Omega^L$ on the (generalized) distribution $\rho(L)$. Indeed, given
$X_x+\bar{\bar{\alpha}}_x\in L$, we have that ${\alpha}_x-\bar{\bar{\alpha}}_x\in
L$. Since $L$ is isotropic, $\langle {\alpha}_x-\bar{\bar{\alpha}}_x,
\bar{X}_x+\bar{\alpha}_x\rangle_+=({\alpha}_x-\bar{\bar{\alpha}}_x)\bar{X}_x/2=0$ for all
$ \bar{X}_x+\bar{\alpha}_x\in L$. Then,  $[\Omega_x^L(X_x)](\bar X_x)=-\bar{\bar
{\alpha}}_x( \bar{X}_x)=-{\alpha}_x(\bar{X}_x)$ for all $\bar{X}_x\in \rho(L)$ and
$\Omega^L$ is well defined.

It is easy to see that gauge vector fields generate the {\it gauge distribution} $\ker \Omega^L$. Moreover, the involutivity of $L$ ensures that $\rho(L)$ is an integrable
generalized distribution in the sense of Stefan--Sussmann \cite{Su73}.
Therefore, it induces a (generalized) foliation $\mathfrak{F}^L=\{\mathfrak{F}^L_x: x\in N\}$ on $N$.

Since $\rho(L_x)=T_x\mathfrak{F}_x^L$, if the elements $X_x+\alpha_x$ and $X_x+\bar{\alpha}_x$, with $X_x\in T_x\mathfrak{F}_x^L$, are in $L_x\subset \mathcal{P}_xN=T_xN\oplus T^*_xN$, then $\alpha_x-\bar{\alpha}_x$ is in the annihilator of $T_x\mathfrak{F}_x^L$, so the image of $\alpha_x$ under the canonical restriction
$\sigma:\alpha_x\in T^*_xN\mapsto \alpha_x|_{T_x\mathfrak{F}_x^L}\in T^*_x\mathfrak{F}_x^L$ is uniquely determined. One can verify that $\sigma(\alpha_x)=-\Omega^L_x(X_x)$. The two-form $\Omega^L$ restricted to $\mathfrak{F}_x^L$ turns out to be closed, so that $\mathfrak{F}_x^L$ is canonically a presymplectic manifold, and the canonical restriction of $L$ to $\mathfrak{F}_x^L$ is the graph of this form \cite{Co87}.

As particular instances,
Poisson and presymplectic
manifolds are particular cases of Dirac
manifolds. On one hand, consider a presymplectic manifold $(N,\omega)$ and
define $L^\omega$ to be the graph of minus the fiber bundle morphism
$\widehat\omega:X_x\in TN\mapsto \omega_x(X_x,\cdot) \in T^*N$. The generalized
distribution $L^\omega$ is isotropic, as
\[
\langle
X_x-\widehat {\omega}(X_x),\bar{X}_x-\widehat {\omega}(\bar{X}_x)\rangle_+=-({\omega}_x(X_x,\bar{X}_x)+\omega_x(\bar{X}_x,X_x))/2=0\,.
\]
As $L^\omega$ is the graph of $-\widehat{\omega}$, then $\dim L^\omega_x=\dim N$ and $L^\omega$ is a maximally
isotropic subbundle of $\mathcal{P}N$. In addition, its integrability relative to
the Courant--Dorfman bracket comes from the fact that $d\omega=0$. Indeed, for arbitrary $X,X'\in \Gamma(TN)$, we have
\[
[[X-\iota_X\omega,X'-\iota_{X'}\omega]]_C=[X,X']-\mathcal{L}_X \iota_{X'}\omega+\iota_{X'}d\iota_{X}\omega=
[X,X']-\iota_{[X,X']}\omega\,,
\]
since
\[
\mathcal{L}_X \iota_{X'}\omega-\iota_{X'}d\iota_{X}\omega=\mathcal{L}_X \iota_{X'}\omega-\iota_{X'}\mathcal{L}_{X}\omega=\iota_{[X,X']}\omega\,.\]
In this case, $\rho:L^\omega\rightarrow TN$ is a bundle isomorphism.
Conversely, given a
Dirac manifold whose $\rho:L\rightarrow TN$ is a bundle isomorphism, its characteristic distribution satisfies $\rho(L)=TN$ and it admits a unique integral leaf, namely
$N$, on which $\Omega^L$ is a closed two-form, i.e., $(N,\Omega^L)$ is a presymplectic manifold.

On the other hand, every Poisson manifold $(N,\Lambda)$ induces a subbundle
 $L^\Lambda$ given by the graph of $\widehat\Lambda$. It is isotropic,
\[
\langle \widehat{\Lambda}(\alpha_x)+ \alpha_x,\widehat{\Lambda}(\bar{\alpha}_x)+
\bar{\alpha}_x\rangle_+=(\Lambda_x(\bar{\alpha}_x,\alpha_x)+\Lambda_x(\alpha_x,
\bar{\alpha}_x))/2=0,
\]
for all $ \alpha_x,\bar{\alpha}_x\in T^*_xN$  and $x\in N$,
and of rank $\dim N$ as the graph of $\widehat \Lambda$ is a map from $T^*N$. Additionally, $L^\Lambda$ is integrable. Indeed, as $\widehat \Lambda (d\{f,g\})=[\widehat \Lambda(df),\widehat \Lambda(dg)]$ for every $f,g\in C^\infty(N)$ \cite{IV}, we have
\[
\![[\widehat \Lambda(df)\!+\!df,\!\widehat \Lambda(dg)+dg]]_C\!\!=\!\![\widehat \Lambda(df),\!\widehat \Lambda(dg)]+\mathcal{L}_{\widehat \Lambda(df)}dg-\iota_{\widehat \Lambda(dg)}d^2\!f\!\!=\!\!
\widehat \Lambda(d\{f,g\})+d\{f,g\}
\]
and the involutivity follows from the fact that the module of 1-forms is generated locally by exact 1-forms.

Conversely, every Dirac manifold $(N,L)$ such
that $\rho^*:L\rightarrow T^*N$ is a bundle isomorphism is the graph of $\widehat\Lambda$ of a Poisson bivector.

Let us motivate our terminology. We call  $\rho(L)$ the
 characteristic distribution of $(N,L)$, which follows the terminology of \cite{NTZ12}
instead of the original one by Courant \cite{Co87}. This is done because when
$L$ comes from a Poisson manifold, $\rho(L)$ coincides with the characteristic
distribution of the Poisson structure \cite{IV}. Meanwhile, the vector fields taking values in $\ker \Omega^L$ are called {\it gauge vector fields}. In this way, when $L$ is the graph of a presymplectic
structure, such vector fields are its gauge vector fields
\cite{EMR99}.

\section{Actions, momentum maps, and invariants on Dirac manifolds}
In the standard symplectic setting, momentum maps are associated with Hamiltonian actions of Lie groups. We will present an analogous concept for Hamiltonian actions on Dirac manifolds, however, limiting ourselves to infinitesimal actions which is sufficient for the theory and our purposes.
\begin{definition}
Let us assume that $(N,L)$ is a Dirac manifold equipped with an {\it infinitesimal $L$-Hamiltonian action} of a finite-dimensional real Lie algebra $\mathfrak{g}$, i.e., a Lie algebra homomorphism $\phi:(\mathfrak{g},[\cdot,\cdot])\rightarrow ({\rm Adm}(N,L),\{\cdot,\cdot\}_L)$. The {\it momentum map} associated with $\phi$ is the map $J_\phi:N\to\mathfrak{g}^*$ defined by
\begin{equation*}
[J_\phi(x)](v)=[\phi(v)](x)\,,\qquad \forall v\in\mathfrak{g},\,\, \forall x\in N.
\end{equation*}
\end{definition}
\noindent Note that $\mathfrak{g}^*$ is canonically a Poisson manifold with respect to the Kirillov--Konstant--Souriau Poisson structure for which linear functions $f_v:\theta\in \mathfrak{g}^*\mapsto \theta(v)\in \mathbb{R}$ and $f_w:\theta\in \mathfrak{g}^*\mapsto \theta(w)\in \mathbb{R}$ associated with $v,w\in\mathfrak{g}$ commute as $v,w$ in the Lie algebra $\mathfrak{g}$:
\begin{equation}\label{KKS}\{f_v,f_w\}_{\mathfrak{g}^*}=f_{[v,w]},\,\qquad \forall v,w\in\mathfrak{g}.
\end{equation}
\begin{proposition}\label{invariants} The map
\begin{equation}\label{momentum}
J_\phi^*:f\in C^\infty(\mathfrak{g}^*)\mapsto f\circ J_\phi\in  C^\infty(N)\,
\end{equation}
takes values in ${\rm Adm}(N,L)$ and establishes a morphism of Poisson algebras. In particular,  if $C\in C^\infty(\mathfrak{g}^*)$ is a Casimir function, i.e., a central element in $(\mathfrak{g}^*,\{\cdot,\cdot\})$, then
$C\circ J_\phi$ commutes with all elements of $\phi(\mathfrak{g})$, thus it is an invariant, i.e., a first integral of all Hamiltonian vector fields $X_h$, with $h\in\phi(\mathfrak{g})$.
\end{proposition}
\begin{proof} If $f\!:\!\mathbb{R}^n\to\mathbb{R}$ is a smooth function and $f_1,\dots,f_n\in{\rm Adm}(N,L)$, then $f=f(f_1,\dots,f_n)\in {\rm Adm}(N,L)$. Indeed,
by defining the vector field
\[X=\sum_{i=1}^n\frac{\partial f}{\partial x^i}(f_1,\dots,f_n)X_{f_i}\,,\]
we see that
\[
\begin{gathered}
X+df=\sum_{i=1}^n\frac{\partial f}{\partial x^i}(f_1,\dots,f_n)X_{f_i}+\sum_{i=1}^n\frac{\partial f}{\partial x^i}(f_1,\dots,f_n)df_i=\\\sum_{i=1}^n\frac{\partial f}{\partial x^i}(f_1,\dots,f_n)(X_{f_i}+df_i).
\end{gathered}
\]
Since $L$ is a vector bundle and $X_{f_i}+df_i$ belong to $\Gamma(L)$, then $X+df\in \Gamma(L)$ and $f$ becomes an $L$-Hamiltonian, i.e., admissible, function for $X$. The rest easily follows from the fact that the momentum map is Poisson, namely if $v,w\in\mathfrak{g}$, then
\[
\begin{gathered}
\{J^*_\phi f_v,J^*_\phi f_w\}_L=\{f_v\circ J_\phi,f_w\circ J_\phi\}_L=\{\phi(v),\phi(w)\}_L=\\\phi([v,w])=f_{[v,w]}\circ J_\phi=J_\phi^*\left(\{f_{v},f_{w}\}_{\mathfrak{g}^*}\right).
\end{gathered}
\]
In particular,
\[X_{\phi(v)}(J^*_\phi(C))=\{\phi(v),J^*_\phi(C)\}_L=J_\phi^*\left(\{f_v,C\}_{\mathfrak{g}^*}\right)=0\,.\]
\end{proof}

\section{Lie systems, Lie--Hamilton systems, and related notions}

We  denote a real Lie algebra by a pair
 $(V, [ \cdot, \cdot])$, i.e.  $V$
stands for a real linear space endowed with a Lie bracket $[\cdot  \hspace{0.25em},
\cdot ] : V \times V \rightarrow V$. Given two subsets $\mathcal{A},
\mathcal{B} \subset V$, we write $[ \mathcal{A}, \mathcal{B}]$ for the real
linear space spanned by the Lie brackets between elements of $\mathcal{A}$ and
$\mathcal{B}$, and we define Lie$( \mathcal{B}, V, [ \cdot,
\cdot])$ to be the smallest Lie subalgebra of $V$ containing $\mathcal{B}$.
When their meaning is clear, we use Lie$( \mathcal{B})$ and $V$ to represent
Lie$( \mathcal{B}, V, [ \cdot, \cdot])$ and $(V, [ \cdot, \cdot])$,
respectively.

A {\it $t$-dependent vector field} on $N$ is a map $X : (t, x)
\in
\mathbb{R} \times N \mapsto X (t, x) \in TN$ such that $\tau_N \circ X =
\pi_2$, where $\pi_2 : (t, x) \in \mathbb{R} \times N \mapsto x \in N$ and
$\tau_N : TN \rightarrow N$ is the tangent bundle projection associated to $N$. This
condition entails that $X$ amounts to a family of
vector fields $\{X_t \}_{t \in \mathbb{R}}$, with $X_t : x \in N \mapsto X (t,
x) \in TN$ for all $t \in \mathbb{R}$ and vice versa {\cite{Dissertationes}}.
The {\it minimal Lie algebra} of $X$  is the
  smallest real Lie algebra, $V^X$, containing the vector fields $\{X_t \}_{t
  \in \mathbb{R}}$, namely $V^X ={\rm Lie} (\{X_t \}_{t \in \mathbb{R}})$.

Any {\it integral curve} of $X$ corresponds to an integral curve  $\gamma:
\mathbb{R} \mapsto \mathbb{R} \times N$ of the {\it
suspension} of $X$,
i.e. the vector field $\partial / \partial t+X (t, x)$ on $\mathbb{R} \times
N$ {\cite{FM}}. Every integral curve $\gamma$ of the form $t\mapsto (t,x(t))$
satisfies
\[ \frac{d (\pi_2 \circ \gamma)}{d t} ( t) = (X \circ \gamma) (
   t) . \]
This system is referred to as the {\it associated system} of $X$.
Conversely, every system of first-order differential equations in the normal form
describes the integral curves $(t,x(t))$ of a unique $t$-dependent vector field.
This
establishes a bijection between $t$-dependent vector fields and systems of
first-order differential equations in the normal form, which justifies the use of $X$
to denote both: the $t$-dependent vector field and its associated system.

The {\it associated distribution} of a  $t$-dependent vector field $X$ on $N$ is the generalized distribution $\mathcal{D}^X$ on $N$ spanned by the
  vector fields of $V^X$, i.e.
  \[ \mathcal{D}^X_x = \{Y_x \mid Y \in V^X \} \subset T_x N. \]

Observe that {r}$^X : x \in N \mapsto \dim
\mathcal{D}^X_x \in \mathbb{N} \cup \{0\}$ needs not to be constant on $N$. We
can only guarantee that {r}$^X (x) = k$ implies {r}$^X (x')
\geq {\rm r}^X (x)$ for $x'$ in a neighborhood of $x$. It follows
that {r}$^X$ is a {\it lower semicontinuous function}
which is constant on the connected components of an open and dense subset
$U^X$ of $N$ (cf. {\cite[p. 19]{IV}}), where $\mathcal{D}^X$ becomes a regular
involutive distribution. The most relevant instance for us is when $\mathcal{D}^X$ is determined by a
finite-dimensional $V^X$ and hence $\mathcal{D}^X$ is integrable (in the sense of Stefan--Sussmann) on $N$ {\cite[p. 63]{JPOT}}.

Let us now turn to some fundamental notions appearing in the theory of Lie
systems.

\begin{definition}{\bf (Vessiot 1893 \cite{Ve93})}  A {\it Lie system} is a
system of
first-order ordinary differential equations $X$ on a manifold $N$ such that
$X_t=\sum_{k=1}^rb_k(t)X_k,$
for a certain collection of $t$-dependent functions $b_1,\ldots,b_r$ and a
family of $t$-independent vector fields $X_1,\ldots,X_r$ on $N$ spanning an
$r$-dimensional real Lie algebra of vector fields.
\end{definition}

Following the terminology in \cite{Ib00,Dissertationes}, we call the real Lie algebra
spanned by $X_1,\ldots,X_r$ a {\it Vessiot--Guldberg Lie algebra} for $X$. Its importance is due to its use in devising various
methods of integration of Lie systems
\cite{BCHLS13,CGM00,CGM07,CarRamGra,CLS12Ham}, especially
in the derivation of superposition rules \cite{CGM07}, which allow us to reduce the integration of a
Lie system to deriving finite families of particular solutions.

\begin{definition}
  A {\it superposition rule} depending on $m$ particular solutions for a system
$X$
  on $N$ is a function $\Phi : N^m \times N \rightarrow N$, $x = \Phi
  (x_{(1)}, \ldots, x_{(m)} ; \lambda)$, such that the general solution $x
  (t)$ of $X$ can be brought into the form $x (t) = \Phi (x_{(1)} (t), \ldots,
  x_{(m)} (t) ; \lambda)$, where $x_{(1)} (t), \ldots, x_{(m)} (t)$ is any
  generic family of particular solutions and $\lambda$ is a point of $N$ to be
  related to initial conditions.
\end{definition}

The conditions ensuring that a system $X$ possesses a superposition rule are
stated in the {\it Lie--Scheffers Theorem} (see \cite{BM09II,CGM07,LS} for details).

\begin{theorem} {\bf (Lie--Scheffers Theorem)}
  A system $X$ admits a superposition rule if and only if $X$ is a Lie system.
Equivalently, $X$ possesses a superposition rule if and only if $V^X$ is
finite-dimensional.
\end{theorem}

The simplest nonlinear example of a Lie system is the Riccati equation, i.e.
\begin{equation}\label{Riccati}
\frac{dx}{dt}=a_0(t)+a_1(t)x+a_2(t)x^2,
\end{equation}
where $a_0(t)$, $a_1(t)$ and $a_2(t)$ are arbitrary $t$-dependent functions \cite{PW}. This equation is related to the $t$-dependent vector field
\[
X_t=(a_0(t)+a_1(t)x+a_2(t)x^2)\frac{\partial}{\partial x},
\]
which can be written as the linear combination
$X_t=a_0(t)X_1+a_1(t)X_2+a_2(t)X_3$, with
\[
X_1=\frac{\partial}{\partial x},\qquad X_2=x\frac{\partial }{\partial x},\qquad
X_3=x^2\frac{\partial }{\partial x},
\]
which satisfy the commutation relations
\[
[X_1,X_2]=X_1,\qquad [X_1,X_3]=2X_2,\qquad [X_2,X_3]=X_3,
\]
and therefore spanning a finite-dimensional real Lie algebra of vector fields
isomorphic to $\mathfrak{sl}(2,\mathbb{R})$ \cite{Dissertationes,PW}. A
superposition rule for Riccati equations is given by the function
$\Phi:\mathbb{R}^3\times \mathbb{R}\rightarrow \mathbb{R}$ of the form \cite{PW}
\[
\Phi(u_{(1)},u_{(2)},u_{(3)};\lambda)=\frac{u_{(1)}(u_{(2)}-u_{(3)})-\lambda u_{(2)}(u_{(3)}
-u_{(1)})}{(u_{(2)}-u_{(3)})-\lambda(u_{(3)}-u_{(1)})},
\]
which allows us to recover the general solution, $x(t)$, of any Riccati equation in terms
of three different particular solutions, $x_{(1)}(t),x_{(2)}(t),x_{(3)}(t)$, and
a real constant $\lambda$ as follows:
\[
x(t)=\Phi(x_{(1)}(t),x_{(2)}(t),x_{(3)}(t);\lambda).
\]

One can devise more powerful methods to study Lie systems admitting particular types of Vessiot--Guldberg Lie algebras \cite{CLS12Ham}, e.g.,  to
consider Lie--Hamilton systems \cite{ADR11,CLS12Ham}.

\begin{definition}
A {\it Lie--Hamilton system} $X$ on $N$ is a Lie system that possesses a
Vessiot--Guldberg Lie algebra of Hamiltonian vector fields with respect to a
Poisson bivector on $N$.
\end{definition}

As expected, imposing additional conditions on the Vessiot--Guldberg Lie
algebras reduces substantially classes of Lie systems we can consider \cite{BHLS12Cla}. For instance, most Riccati equations are not Lie--Hamilton systems: they are defined on the real line and $X=0$ is the unique Lie--Hamilton system on the real line. More specifically, every Poisson bivector on $\mathbb{R}$ is null and, recalling
that every Hamiltonian vector field can be written in the form
$X=\widehat\Lambda(df)$ for a Hamiltonian function $f$, we see that
every Vessiot--Guldberg Lie algebra of Hamiltonian vector fields on $\mathbb{R}$
is trivial. So, given a Lie--Hamilton system on the real line we have that $V^X=0$ and $X=0$. Despite this, Lie--Hamilton systems can be applied to study relevant differential equations, e.g.,
second-order Riccati equations, second-order Kummer--Schwarz equations
\cite{CLS12Ham}, and in a particular type of mechanical systems generated by Lie
algebras of functions \cite{BBHMR09,BBR06,BO98}.

The main feature of a Lie--Hamilton system is the following property \cite{CLS12Ham}.
\begin{theorem}\label{BasicLieHam}
Given a Lie--Hamilton system $X$ admitting a Vessiot--Guldberg Lie algebra of Hamiltonian vector fields relative to a Poisson manifold $(N,\Lambda)$,  there exists
a $t$-dependent Hamiltonian $h:(t,x)\in \mathbb{R}\times N\mapsto h_t(x)=h(t,x)\in
\mathbb{R}$ such that $X_t=\widehat{\Lambda}(dh_t)$, for every $t\in\mathbb{R}$, and the functions
$\{h_t\}_{t\in\mathbb{R}}$ span a finite-dimensional real Lie algebra  with respect to the Poisson structure $\{\cdot,\cdot\}_\Lambda$ induced by $(N,\Lambda)$.
\end{theorem}

In the latter case, we say that $X$ admits a {\it Lie--Hamiltonian structure} $(N,\Lambda,h)$. The following observations are immediate consequences of the corresponding definitions (see \cite{BCHLS13,CLS12Ham} for details).

\begin{proposition}\label{Prop:Int} Let us assume that $X$ is a Lie--Hamilton system with
a Lie--Hamiltonian $(N,\{\cdot,\cdot\},h)$. Then, $f$ is a $t$-independent
constant of the motion for $X$ if and only if it Poisson commutes with the elements
of ${\rm Lie}(\{h_t\}_{t\in\mathbb{R}},\{\cdot,\cdot\})$. The space
$\mathcal{I}^X$ of $t$-independent constants of the motion for $X$ is a Poisson
algebra $(\mathcal{I}^X,\sbt,\{\cdot,\cdot\})$. If $f$ is a $t$-independent constant of
motion for $X$, then the Hamiltonian vector field associated to $f$ is a
$t$-independent Lie symmetry for $X$.
\end{proposition}

In order to illustrate the above notions, let us provide a
simple example of mathematical and physical interest: the system of Riccati equations
\begin{equation}\label{CoupledRic}
\begin{aligned}
\frac{dx_i}{dt}&=a_0(t)+a_1(t)x_i+a_2(t)x_i^2,\qquad i=1,\ldots,4,\\
\end{aligned}
\end{equation}
where $a_0(t),a_1(t),a_2(t)$ are arbitrary $t$-dependent functions and we assume
that $(x_1-x_2)(x_2-x_3)(x_3-x_4)\neq 0$.
The determination of a common $t$-independent constant of the motion $F$ for all systems of this type
leads to deriving a superposition rule for Riccati equations (cf. \cite{CGM07}).
The standard methods to derive $F$ require the integration of a system of PDEs \cite{CGM07} or ODEs \cite{PW}. Nevertheless, we next show that, since  (\ref{CoupledRic}) is a Lie--Hamilton system, we can obtain $F$ from algebraic manipulations without integrating any system of  PDE or ODEs \cite{BCHLS13}.

System (\ref{CoupledRic}) determines integral curves of the $t$-dependent vector field
\begin{equation}\label{Field}
X_{R}=\sum_{i=1}^4(a_0(t)+a_1(t)x_i+a_2(t)x_i^2)\frac{\partial}{\partial x_i}=a_0(t)X_1+a_1(t)X_2+a_2(t)X_3,
\end{equation}
with
\[
X_1=\sum_{i=1}^4\frac{\partial}{\partial x_i},\qquad X_2=\sum_{i=1}^4x_i\frac{\partial}{\partial x_i},\qquad X_3=\sum_{i=1}^4x_i^2\frac{\partial}{\partial x_i}.
\]
Using that $X_1, X_2,$ and $X_3$ span a finite-dimensional real Lie algebra of vector fields, we see that $X_R$ is a Lie system.
Consider the two-form
\begin{equation}\label{SymCoupledRic}
\omega_R=\frac{dx_1\wedge dx_2}{(x_1-x_2)^2}+\frac{dx_3\wedge dx_4}{(x_3-x_4)^2}.
\end{equation}
Note that $\omega_R$ is a symplectic form on $\mathcal{O}=\{(x_1,x_2,x_3,x_4)|(x_1-x_2)(x_2-x_3)(x_3-x_4)\neq 0\}\subset \mathbb{R}^4$ and
\[
\begin{gathered}
\iota_{X_1}\omega_R=d\left(\frac{1}{x_1-x_2}+\frac{1}{x_3-x_4}\right),\,\,
\qquad \iota_{X_2}\omega_R=
\frac 12d\left(\frac{x_1+x_2}{x_1-x_2}+\frac{x_3+x_4}{x_3-x_4}\right),\\
\iota_{X_3}\omega_R=d\left(\frac{x_1 x_2}{x_1-x_2}+\frac{x_3 x_4}{x_3-x_4}\right).
\end{gathered}
\]
Hence, the vector fields $X_1,X_2,$ and $X_3$ are Hamiltonian with respect to $(\mathcal{O},\omega_R)$ with the Hamiltonian functions
\[
\begin{gathered}
h_1=-\frac{1}{x_1-x_2}-\frac{1}{x_3-x_4},\quad
h_2=-\frac 12\left(\frac{x_1+x_2}{x_1-x_2}+\frac{x_3+x_4}{x_3-x_4}\right),
\\h_3=-\frac{x_1 x_2}{x_1-x_2}-\frac{x_3 x_4}{x_3-x_4},
\end{gathered}
\]
respectively. This shows that system (\ref{CoupledRic}) is a Lie--Hamilton system. Additionally,
\[
\{h_1,h_2\}=h_1,\qquad \{h_1,h_3\}=2h_2,\qquad \{h_2,h_3\}=h_3,
\]
where $\{\cdot,\cdot\}$ stands for the natural Poisson bracket induced by the symplectic form $\omega_R$, and ${\rm Lie}(\{h_1,h_2,h_3\},\{\cdot,\cdot\})\simeq \mathfrak{sl}(2,\mathbb{R})$. It is known that $h_1$, $h_2,$ and $h_3$ Poisson commute with $h_1h_3-h_2^2$, which can be considered, up to a constant factor, as a Casimir function of the Poisson manifold $C^\infty(\mathfrak{sl}(2,\mathbb{R})^*)$ (see \cite{BBHMR09} for details). In other words,
\[
C=h_1h_3-h_2^2=\frac{(x_2-x_3)(x_1-x_4)}{(x_1-x_2)(x_3-x_4)}
\]
Poisson commutes with $h_1,h_2$ and $h_3$. Using that $0=\{h_k,C\}=X_k C$ for $k=1,2,3$, we see that $C$
is a common first-integral for $X_1,X_2,$ and $X_3$. From this, it turns out that $C$ is a $t$-independent constant of the motion for $X$.
Observe that the above method can be applied to other Lie--Hamilton systems {\it mutatis mutandis}. The works \cite{ADR11,BCHLS13,CLS12Ham} include many other techniques that can be applied to these systems.
 We will see in this work that this kind of procedures can be applied to even more general types of systems.

\section{On the necessity of Dirac--Lie systems}

Many systems have recently been found to be Lie--Hamilton systems \cite{ADR11,BCHLS13,CLS12Ham}. This permitted us to use several geometric and algebraic techniques to study their superposition rules, constants of the motion, and Lie symmetries. Despite the advantages of these methods, they are not applicable to all Lie systems, as some of them do not
admit any Vessiot--Guldberg Lie algebra of Hamiltonian vector fields.
Let us illustrate this through several examples.

Consider a third-order
Kummer--Schwarz equation \cite{Be88,GL12} of the form
\begin{equation}\label{KS3}
\frac{d^3x}{dt^3}=\frac 32\left(
\frac{dx}{dt}\right)^{-1}\!\!\left(\frac{d^2x}{dt^2}\right)^{2}\!\!-2c_0\left(\frac{
dx}{dt}\right)^3\!\!+2b_1(t)\frac{dx}{dt},
\end{equation}
where $c_0$ is a real constant and $b_1=b_1(t)$ is any
$t$-dependent function.  This differential equation is known to be a HODE Lie system \cite{CGL11}. This means that the system of first-order differential equations obtained by adding the variables $v\equiv dx/dt$ and $a\equiv d^2x/dt^2$, namely
\begin{equation}\label{firstKS3}
\frac{dx}{dt}=v,\qquad \frac{dv}{dt}=a,\qquad \frac{da}{dt}=\frac 32 \frac{a^2}v-2c_0v^3+2b_1(t)v\,,
\end{equation}	
 is a Lie system. Indeed, it is associated to the $t$-dependent vector field
\begin{equation}\label{Ex}
X^{3KS}_t=v\frac{\partial}{\partial x}+a\frac{\partial}{\partial v}+\left(\frac 32
\frac{a^2}v-2c_0v^3+2b_1(t)v\right)\frac{\partial}{\partial
a}=Y_3+b_1(t)Y_1,
\end{equation}
where the vector fields on $\mathcal{O}_2=\{(x,v,a)\in{T}^2\mathbb{R}\mid v\neq 0\}$ given by
\begin{equation}\label{VFKS1}
Y_1=2v\frac{\partial}{\partial a},\quad Y_2=v\frac{\partial}{\partial v}+2a\frac{\partial}{\partial a},\quad Y_3=v\frac{\partial}{\partial x}+a\frac{\partial}{\partial v}+\left(\frac 32
\frac{a^2}v-2c_0v^3\right)\frac{\partial}{\partial a},
\end{equation}
 satisfy
the commutation relations
\begin{equation}
[Y_1,Y_3]=2Y_2,\quad [Y_1,Y_2]=Y_1,\quad [Y_2,Y_3]=Y_3.
\end{equation}
In consequence, $Y_1, Y_2,$ and $Y_3$ span a three-dimensional Lie algebra of vector fields $V$ isomorphic
to $\mathfrak{sl}(2,\mathbb{R})$ and $X^{3KS}$ becomes a $t$-dependent vector field taking values in $V$, i.e. $X^{3KS}$ is a Lie system. However, $X^{3KS}$ is not a Lie--Hamilton system when $b_1(t)$ is not a constant. Indeed, in this case $\mathcal{D}^{X^{3KS}}$ coincides with $T\mathcal{O}_2$ on $\mathcal{O}_2$. If $X^{3KS}$ were also a Lie--Hamilton system with respect to $(N,\Lambda)$, then $V^{X^{3KS}}$ would consist of Hamiltonian vector fields and the characteristic distribution associated to $\Lambda$ would have odd-dimensional rank on $\mathcal{O}_2$. This is impossible, as the local Hamiltonian vector fields of a Poisson manifold span a generalized distribution of even rank at each point. Our previous argument can easily be generalised to formulate the following `no-go' theorem.

\begin{proposition}\label{NoGo} If $X$ is a Lie system on an odd-dimensional manifold $N$ satisfying that $\mathcal{D}^X_{x_0}=T_{x_0}N$ for a point $x_0$ in $N$, then $X$ is not a Lie--Hamilton system on $N$.
\end{proposition}

Note that from the properties of ${\rm r}^X$ it follows that, if $\mathcal{D}^X_{x_0}=T_{x_0}N$ for a point $x_0$, then $\mathcal{D}^X_{x}=T_{x}N$ for $x$ in an open neighborhood $U_{x_0}\ni x_0$. Hence, we can merely consider whether $X$ is a Lie--Hamilton system on $N\backslash U_{x_0}$.

Despite the previous negative results, system (\ref{firstKS3}) admits another interesting property:  we can endow the
manifold $\mathcal{O}_2$ with a presymplectic form $\omega_{3KS}$ in such a way that $V^{X^{3KS}}$ consists of Hamiltonian vector fields with respect to it. Indeed, by considering the equations $
\mathcal{L}_{Y_1}\omega_{3KS}=\mathcal{L}_{Y_2}\omega_{3KS}=\mathcal{L}_{Y_3}\omega_{3KS}=0$ and $d\omega_{3KS}=0$, we can readily find the presymplectic form
\[
\omega_{3KS}=\frac{dv\wedge da}{v^3}
\]
on $\mathcal{O}_2$. Additionally, we see that
\begin{equation}\label{3KSHamFun}
\iota_{Y_1}\omega_{3KS}=d\left(\frac{2}{v}\right),\qquad \iota_{Y_2}\omega_{3KS}=d\left(\frac{a}{v^2}\right), \qquad \iota_{Y_3}\omega_{3KS}=d\left(\frac{a^2}{2v^3}+2c_0v\right).
\end{equation}
So, the system $X^{3KS}$ becomes a Lie system with a Vessiot--Guldberg Lie algebra of Hamiltonian vector fields
with respect to $\omega_{3KS}$. As seen later on, systems of this type can be studied through
appropriate generalizations of the methods employed to investigate Lie--Hamilton systems.

Another example of a Lie system which is not a Lie--Hamilton system but admits a Vessiot--Guldberg Lie algebra of Hamiltonian vector fields with respect to a presymplectic form is the {\it Riccati system}
\begin{equation}\label{Partial}
\left\{
\begin{aligned}
\frac{ds}{dt}&=-4a(t)u s-2d(t)s,\qquad &\frac{dx}{dt}&=(c(t)+4a(t)u)x+f(t)-2u g(t), \\
\frac{du}{dt}&=-b(t)+2c(t) u+4a(t)u^2,\qquad &\frac{dy}{dt}&=(2a(t)x-g(t))v,\\
\frac{dv}{dt}&=(c(t)+4a(t)u)v,\qquad &\frac{dz}{dt}&=a(t)x^2-g(t)x, \\
\frac{dw}{dt}&=a(t)v^2, &\\
\end{aligned}\right.
\end{equation}
where $a(t)$, $b(t),$  $c(t)$, $d(t)$, $f(t)$ and $g(t)$ are arbitrary $t$-dependent functions. The interest of this system is due to its use in solving diffusion-type equations, Burger's equations, and other PDEs \cite{SSVG11}.

Taking into account that every particular solution $(s(t),u(t),v(t),w(t),x(t),y(t),z(t))$ of (\ref{Partial}), with $v(t_0)=0$ ($s(t_0)=0$) for a certain $t_0\in\mathbb{R}$, satisfies
$v(t)=0$ ($s(t)=0$) for every $t$, we can restrict ourselves to analyzing system (\ref{Partial}) on the submanifold $M=\{(s,u,v,w,x,y,z)\in\mathbb{R}^7\,|\,v\neq 0, s\neq 0\}$. This will simplify the application of our  techniques without omitting any relevant detail.

System (\ref{Partial}) describes integral curves of the $t$-dependent vector field
\[
X^{RS}_t=a(t)X_1-b(t)X_2+c(t)X_3-2d(t)X_4+f(t)X_5+g(t)X_6,
\]
where
\[
\begin{gathered}
 X_1=-4us\frac{\partial}{\partial s}+4u^2\frac{\partial}{\partial u}+4uv\frac{\partial}{\partial v}+v^2\frac{\partial}{\partial w}+4ux\frac{\partial}{\partial x}+2xv\frac{\partial}{\partial y}+x^2\frac{\partial}{\partial z},\\  X_2=\frac{\partial}{\partial u},\quad X_3=2u\frac{\partial}{\partial u}+v\frac{\partial }{\partial  v}+x\frac{\partial}{\partial x}, \quad X_4=s\frac{\partial}{\partial s},\quad  X_5=\frac{\partial}{\partial x},\\ X_6=-2u\frac{\partial}{\partial x}-v\frac{\partial}{\partial y}-x\frac{\partial}{\partial z},\quad X_7=\frac{\partial}{\partial z}.
\end{gathered}
\]
Their commutation relations are
\[
\begin{gathered}
\left[X_1,X_2\right]=4(X_4 -X_3), \quad [X_1,X_3]=-2X_1,\quad [X_1,X_5]=2X_6, \quad [X_1.X_6]=0,\\
[X_2,X_3]=2X_2, \quad [X_2,X_5]=0,\quad [X_2,X_6]=-2X_5, \\
 [X_3,X_5]=-X_5, \quad [X_3,X_6]=X_6,\\
[X_5,X_6]=-X_7, \\
\end{gathered}
\]
and $X_4$ and $X_7$ commute with all the vector fields.
Hence, system (\ref{Partial}) is a Lie system associated to a Vessiot--Guldberg Lie algebra $V$ isomorphic to $(\mathfrak{sl}(2,\mathbb{R})\ltimes \mathfrak{h_2})\oplus \mathbb{R}$, where $\mathfrak{sl}(2,\mathbb{R})\simeq \langle X_1,X_2,X_4-X_3\rangle$, $\mathfrak{h_2}\simeq \langle X_5,X_6,X_7\rangle$ and $\mathbb{R}\simeq \langle X_4\rangle$.
It is worth noting that this new example of Lie system is one of the few Lie systems related to remarkable PDEs until now \cite{CGM07}.

Observe that (\ref{Partial}) is not a Lie--Hamilton system when $V^{X^{RS}}=V$. In this case $\mathcal{D}^{X^{RS}}_p=T_pM$ for any $p\in M$ and, in view of Proposition \ref{NoGo} and the fact  that $\dim T_pM=7$, the system $X^{RS}$ is not a Lie--Hamilton system on $M$.

Nevertheless, we can look for a presymplectic form turning $X^{RS}$ into a Lie system with a Vessiot--Guldberg Lie algebra of Hamiltonian vector fields. Looking for a non-trivial solution of the system of equations $\mathcal{L}_{X_k}\omega_{RS}=0$, with $k=1,\ldots,7$, and $d\omega_{RS}=0$, one can find the presymplectic two-form
\[
\omega_{RS}=-\frac{4wdu\wedge dw}{v^2}+\frac{dv\wedge dw}{v}+\frac{4w^2du\wedge dv}{v^3}.
\]
In addition, we can readily see that $d\omega_{RS}=0$ and $X_1,\ldots,X_r$ are Hamiltonian vector fields:
\begin{equation}\label{HamPar}
\iota_{X_1}\omega_{RS}=d\left(4uw-\frac{8u^2w^2}{v^2}-\frac{v^2}{2}\right),\quad \iota_{X_2}\omega_{RS}=-d\left(\frac{2w^2}{v^2}\right),\\  \iota_{X_3}\omega_{RS}=d\left(w-\frac{4w^2u}{v^2}\right),
\end{equation}
and $\iota_{X_k}\omega_{RS}=0$ for $k=4,\ldots,7$.

Apart from the above examples, other non Lie--Hamilton systems that admit a Vessiot--Guldberg Lie algebra of Hamiltonian vector fields with respect to a presymplectic form can be found in the study of certain reduced Ermakov systems \cite{HG04}, Wei--Norman equations for dissipative quantum oscillators \cite{Dissertationes},  and $\mathfrak{sl}(2,\mathbb{R})$--Lie systems \cite{Pi12}.

A straightforward generalization of the concept of a Lie--Hamilton system to Dirac manifolds would be a Lie system admitting a Vessiot--Guldberg Lie algebra $V$ of vector fields for which there exists a Dirac structure $L$ such that $V$ consists of $L$-Hamiltonian vector fields.
Nevertheless, this definition does not make too much sense, as every Lie system is of this type. If $X$ is a Lie system on $N$, the subbundle $L\equiv TN\subset \mathcal{P}N$ gives rise to a Dirac manifold $(N,L)$, where all vector fields $X\in
\Gamma(L)$ are $L$-Hamiltonian with a zero
$L$-Hamiltonian function. Additionally, examples like this do
not provide any additional information about the Lie system. As in the case of
Lie--Hamiltonian systems \cite{CLS12Ham}, we aim at using the Hamiltonian functions related to the
vector fields of $V$ to study the properties of Dirac--Lie systems.
Unfortunately, these functions are zero in the previous trivial example.

In view of the above-mentioned reasons, it only makes nontrivial sense to consider Dirac--Lie
systems as associated to a fixed Dirac structure. Particularly, the notion becomes
useful only when the elements of $V^X$ admit a rich family of $L$-Hamiltonian functions. This leads to the following definition.

\begin{definition}
  A {\it Dirac--Lie system} is a triple $(N, L, X)$, where $(N, L)$ stands for a
  Dirac manifold and $X$ is a Lie system admitting a Vessiot--Guldberg Lie
algebra of
  $L$-Hamiltonian vector fields.
\end{definition}

Recall that every presymplectic manifold $(N,
\omega)$ gives rise to a Dirac manifold $(N,L^\omega)$ whose distribution
$L^\omega$ is spanned by elements
of $\Gamma(TN\oplus_NT^*N)$ of the form $X - \iota_{X} \omega$ with $X\in\Gamma(TN)$. Obviously, this shows
that the Hamiltonian vector fields for $(N,\omega)$ are $L$-Hamiltonian vector
fields relative to $(N,L)$. From here, it
follows that $(\mathcal{O}_2,L^{\omega_{3KS}},X^{3KS})$ and
$(M, L^{\omega_{RS}}, X^{RS})$ are Dirac--Lie systems. Moreover, note that system (\ref{CoupledRic}), which was proved to
be a Lie--Hamilton system, gives also rise to a Dirac--Lie system $(\mathcal{O},L^{\omega_R},X_R)$.

\section{Dirac--Lie Hamiltonians}

In view of Theorem \ref{BasicLieHam}, every Lie--Hamilton system admits a Lie--Hamiltonian. Since Dirac--Lie systems are  generalizations of these systems, it is
natural to
investigate whether Dirac--Lie systems admit an analogous property.

As an example, consider again the third-order Kummer--Schwarz equation in first-order form (\ref{firstKS3}). Remind that $Y_1$, $Y_2$, and $Y_3$ are Hamiltonian vector fields with respect to
the presymplectic manifold $(\mathcal{O}_2,\omega_{3KS})$.  It follows from relations (\ref{3KSHamFun}) that the vector fields $Y_1$, $Y_2,$ and $Y_3$ have Hamiltonian
functions
\begin{equation}\label{Fun3KS}
h_1=-\frac{2}{v},\qquad h_2=-\frac{a}{v^2},\qquad h_3=-\frac{a^2}{2v^3}-2c_0v,
\end{equation}
respectively. Moreover,
\[
\{ h_1, h_3 \}= 2 h_2, \quad \{h_1,
h_2 \} = h_1, \quad \{ h_2, h_3 \}=
h_3,
\]
where $\{\cdot,\cdot\}$ is the Poisson bracket on ${\rm Adm}(\mathcal{O}_2,\omega_{3KS})$  induced by $\omega_{3KS}$. In consequence, $h_1,h_2, $ and $h_3$ span a finite-dimensional real Lie algebra
isomorphic to $\mathfrak{sl}(2,\mathbb{R})$. Thus,
every $X^{3KS}_t$ is a Hamiltonian vector field with Hamiltonian function $
h^{3KS}_t = h_3 + b_1 (t) h_1$ and the space
${\rm Lie}(\{h^{3KS}_t\}_{t\in\mathbb{R}},\{\cdot,\cdot\})$ becomes a
finite-dimensional real Lie algebra. This enables us to associate $X^{3KS}$ to a curve in ${\rm Lie}(\{h^{3KS}_t\}_{t\in\mathbb{R}},\{\cdot,\cdot\})$. The similarity of  $(\mathcal{O}_2,\omega_{3KS},h^{3KS})$
with Lie--Hamiltonians are immediate.

If we now turn to the Riccati system (\ref{Partial}), we will see that we can obtain a similar result. More specifically, relations
(\ref{HamPar}) imply that $X_1,\ldots,X_7$ have Hamiltonian functions
\[
h_1=\frac{(v^2-4 u w)^2}{2v^2}, \qquad h_2=\frac{2\omega^2}{v^2},\qquad h_3=\frac{4w^2u}{v^2}-w,
\]
and $h_4=h_5=h_6=h_7=0$. Moreover, given the Poisson bracket on admissible functions induced by $\omega_{3KS}$, we see that
\[
\{h_1,h_2\}=-4h_3,\qquad \{h_1,h_3\}=-2h_1,\qquad \{h_2,h_3\}=2h_2.
\]
Hence, $h_1,\ldots,h_7$ span a real Lie algebra isomorphic to $\mathfrak{sl}(2,\mathbb{R})$ and, as in the previous case, the $t$-dependent vector fields
$X^{RS}_t$ possess Hamiltonian functions $h^{RS}_t=a(t)h_1-b(t)h_2+c(t)h_3$. Again, we can associate $X^{RS}$ to a curve $t\mapsto h^{RS}_t$ in the finite-dimensional real Lie algebra $({\rm Lie}(\{h^{RS}_t\}_{t\in\mathbb{R}}),\{\cdot,\cdot\})$.

The above examples suggest us the following definition.

\begin{definition}
  A {\it Dirac--Lie Hamiltonian structure} is a triple $(N, L, h)$, where $(N, L)$ stands
  for a Dirac manifold and $h$ represents a $t$-parametrized family of
  admissible functions $h_t : N \rightarrow \mathbb{R}$ such that
  Lie$(\{h_t \}_{t \in \mathbb{R}}, \{ \cdot, \cdot \}_L)$ is a
  finite-dimensional real Lie algebra.
  A $t$-dependent vector field $X$ is said to admit, to have or to possess a
  Dirac--Lie Hamiltonian $(N, L, h)$ if $X_t+dh_t \in \Gamma(L)$ for all $t
  \in \mathbb{R}$.
\end{definition}

\begin{note}
For simplicity, we hereafter call Dirac--Lie Hamiltonian structures Dirac--Lie Hamiltonians.
\end{note}

From the above definition, we see that system (\ref{firstKS3}) related to the third-order Kummer--Schwarz equations
possesses a Dirac--Lie Hamiltonian $(N, L^{\omega_{3KS}},h^{3KS})$ and system (\ref{Partial}), used to analyze diffusion equations,
admits a Dirac--Lie Hamiltonian $(N,L^{\omega_{RS}},h^{RS})$.

Let us analyze the properties of Dirac--Lie structures. Observe first that there may be several systems associated to the same Dirac--Lie Hamiltonian. For instance, the systems $X^{RS}$ and
\[
X^{RS}_2=a(t)X_1-b(t)X_2+c(t)X_3-2d(t)X_4+f(t)z^3X_5+g(t)X_6+h(t) z^2 X_7
\]
admit the same Dirac--Lie Hamiltonian  $(N, L^{\omega_{RS}},h^{RS})$. It is remarkable that $X_2^{RS}$ is not even a Lie system in general. Indeed, in view of
\[
[z^2X_7,z^nX_5]=nz^{n+1}X_5,\qquad n=3,4,\ldots,
\]
we easily see that the successive Lie brackets of $z^nX_5$ and $z^2X_7$ span an infinite set of vector fields which are linearly independent over $\mathbb{R}$. So, in those cases in which $X_5$ and $X_7$ belong to $V^{X^{RS}_2}$, this Lie algebra becomes infinite-dimensional.

In the case of a Dirac--Lie system, Proposition \ref{CasCon} shows easily the following.

\begin{corollary} Let $(N,L,X)$ be a Dirac--Lie system admitting a Dirac--Lie
Hamiltonian $(N,L,h)$. Then, we have
the exact sequence of Lie algebras
\begin{equation*}
0\hookrightarrow {\rm
Cas}(\{h_t\}_{t\in\mathbb{R}},\{\cdot,\cdot\}_L)\hookrightarrow
{\rm
Lie}(\{h_t\}_{t\in\mathbb{R}},\{\cdot,\cdot\}_L)\stackrel{B_L}{\longrightarrow}
\pi(V^X)\rightarrow 0\,,
\end{equation*}
where ${\rm
Cas}(\{h_t\}_{t\in\mathbb{R}},\{\cdot,\cdot\}_L)={\rm Lie}(\{h_t\}_{t\in\mathbb{R}},\{\cdot,\cdot\}_L)\cap {\rm
Cas}(N,L)$.
That is, ${\rm Lie}(\{h_t\}_{t\in\mathbb{R}},\{\cdot,\cdot\}_L)$ is a Lie
algebra extension of $\pi(V^X)$ by ${\rm
Cas}(\{h_t\}_{t\in\mathbb{R}},\{\cdot,\cdot\}_L)$.
\end{corollary}

\begin{theorem}\label{MT}
  Each Dirac--Lie system $(N,L,X)$ admits a Dirac--Lie Hamiltonian $(N,L,h)$.
\end{theorem}

\begin{proof}
  Since $V^X\subset {\rm Ham}(N,L)$ is a finite-dimensional Lie
algebra, we can define
a linear map $T : X_f
  \in V^X \mapsto f \in C^{\infty} (N)$ associating
  each $L$-Hamiltonian vector field in $V^X$ with an associated $L$-Hamiltonian function, e.g., given a basis $X_1,\ldots,X_r$ of $V^X$ we define $T(X_i)=h_i$, with
$i=1,\ldots,r$, and extend $T$ to $V^X$ by linearity. Note that the
  functions $h_1, \ldots, h_r$ need not be linearly independent over $\mathbb{R}$,
  as a function can be Hamiltonian for two different $L$-Hamiltonian vector fields $X_1$ and $X_2$ when $X_1-X_2\in {\rm G}(N,L)$. Given the system $X$,
there exists a smooth curve
  $h_t =T (X_t)$ in $\mathfrak{W}_0\equiv {\rm Im}\,T$ such that $X_t+dh_t \in \Gamma(L)$. To ensure
  that $h_t$ gives rise to a Dirac--Lie Hamiltonian, we need to
  demonstrate that $\dim\, {\rm Lie}(\{h_t \}_{t \in \mathbb{R}}, \{ \cdot,
  \cdot \}_L)<\infty$. This will be done by
constructing a
  finite-dimensional Lie algebra of functions containing the curve $h_t$.

  Consider two elements $Y_1,Y_2\in V^X$. Note that the functions $\{T (Y_1), T (Y_2)\}_L$ and $T ([Y_1, Y_2])$ have the
same
  $L$-Hamiltonian vector field. So, $\{T (Y_1), T (Y_2)\}_L - T ([Y_1,
  Y_2]) \in {\rm Cas}(N,L)$ and, in view of Proposition \ref{CasCon}, it Poisson
commutes with all other admissible functions. Let us define $\Upsilon : V^X
\times V^X \rightarrow C^{\infty} (N)$
  of the form
  \begin{equation}    \label{formula2}
\Upsilon (X_1, X_2) = \{T(X_1), T(X_2)\}_L - T [X_1, X_2].
  \end{equation}
  The image of $\Upsilon$ is contained in a
  finite-dimensional real Abelian Lie subalgebra of Cas$(N,L)$ of the form
  \[ \mathfrak{W}_\mathcal{C} \equiv \langle \Upsilon (X_i, X_j) \rangle,
     \qquad i, j = 1, \ldots, r, \]
  where $X_1, \ldots, X_r$ is a basis for $V^X$. From here, it follows that
  \[ \{ \mathfrak{W}_\mathcal{C}, \mathfrak{W}_\mathcal{C}\}_L = 0,
     \quad \{ \mathfrak{W}_\mathcal{C}, \mathfrak{W}_0 \}_L = 0,
     \quad \{ \mathfrak{W}_0, \mathfrak{W}_0 \}_L \subset
     \mathfrak{W}_\mathcal{C} + \mathfrak{W}_0 . \]
  Hence, $(\mathfrak{W} \equiv \mathfrak{W}_0 + \mathfrak{W}_\mathcal{C},\{\cdot,\cdot\}_L)$ is a
  finite-dimensional real Lie algebra containing the curve $h_t$, and $X$ admits a Dirac--Lie Hamiltonian $(N, L, T(X_t))$.
\end{proof}
The following proposition is easy to check.
\begin{proposition}\label{TanOrb} Let $(N,L,X)$ be a Dirac--Lie system. If $(N,L,h)$ and
$(N,L,\bar h)$ are two Dirac--Lie Hamiltonians for $(N,L,X)$, then
\[
h=\bar h+f^X,
\]
where $f^X\in C^\infty(\mathbb{R}\times N)$ is a $t$-dependent function such
that each $f^X_t:x\in N \mapsto f^X(x,t)\in \mathbb{R}$ is a Casimir function that is constant on every
integral manifold $\mathcal{O}$ of $\mathcal{D}^X$.
\end{proposition}
Note that if we have a Dirac--Lie Hamiltonian $(N,L,h)$ and we define a linear map $\widehat T:  h\in {\rm Lie}(\{h_t\}_{t\in\mathbb{R}},\{\cdot,\cdot\})\mapsto X_h\in {\rm Ham}(N,L)$, the space $\widehat{T}({\rm Lie}(\{h_t\}_{t\in\mathbb{R}},\{\cdot,\cdot\})$ may span an infinite-dimensional Lie algebra of vector fields. For instance, consider again the Lie--Hamiltonian $(\mathcal{O}_2,\omega_{3KS},h^{3KS}_t=h_3+b_1(t)h_1)$ for the system (\ref{firstKS3}). The functions $h_1, h_2,$ and $h_3$ are also Hamiltonian for the vector fields
\[
Y_1=2v\frac{\partial}{\partial a}+e^{v^2}\frac{\partial}{\partial x},\quad Y_2=v\frac{\partial}{\partial v}+2a\frac{\partial}{\partial a},\quad Y_3=a\frac{\partial }{\partial v}+\left(\frac 32 \frac{a^2}v-2c_0v^3\right)\frac{\partial}{\partial a},
\]
which satisfy
\[
\stackrel{j-{\rm times}}{\overbrace{[Y_2,[\ldots,[Y_2}},Y_1]\ldots]]\!=\!f_j(v)\frac{\partial}{\partial x}\!+\!2(-1)^jv\frac{\partial}{\partial a}, \,\,\,\,\,\,\,\,\,\qquad f_j(v)\equiv\,\, \stackrel{j-{\rm times}}{\overbrace{v\frac{\partial}{\partial v}\ldots v\frac{\partial}{\partial v}}}\!\!(e^{v^2}).
\]
In consequence, ${\rm Lie}(\widehat{T}({\rm Lie}(\{h_t\}_{t\in\mathbb{R}},\{\cdot,\cdot\})),[\cdot,\cdot])$ contains an infinite-dimensional Lie algebra of vector fields  because the functions $\{f_j\}_{j\in\mathbb{R}}$ form an infinite family of linearly independent functions over $\mathbb{R}$.  So, we need to impose additional conditions to ensure that the image of $\widehat{T}$ is finite-dimensional.

The following theorem yields an alternative definition of a Dirac--Lie system.

\begin{theorem} Given a Dirac manifold $(N,L)$, the triple $(N,L,X)$ is a
Dirac--Lie system if and only if there exists a curve
$\gamma:t\in\mathbb{R}\rightarrow \gamma_t\in\Gamma(L)$ satisfying that $\rho(\gamma_t)=X_t\in {\rm Ham}(N,L)$ for every $t\in\mathbb{R}$ and
${\rm Lie}(\{\gamma_t\}_{t\in\mathbb{R}},[[\cdot,\cdot]]_C)$ is a
finite-dimensional real Lie algebra.
\end{theorem}
\begin{proof} Let us prove the direct part of the theorem.
 Assume that $(N,L,X)$ is a Dirac--Lie system.
In virtue of Theorem \ref{MT}, it admits a Dirac--Lie Hamiltonian
$(N,L,h)$, with $h_t=T(X_t)$ and $T:V^X\rightarrow {\rm Adm}(N,L)$ a linear morphism associating
each element of $V^X$ with one of its $L$-Hamiltonian functions.
We aim to prove that  the curve in $\Gamma(L)$ of the form $\gamma_t=X_t+d(T(X_t))$ satisfies that $\dim {\rm Lie}(\{\gamma_t\}_{t\in\mathbb{R}},[[\cdot,\cdot]]_C)<\infty$.

The sections of $\Gamma(L)$ of the form
\begin{equation}\label{gen}
X_1+dT(X_1)\,\,,\ldots,\,\,X_r+dT(X_r)\,\,,\,\,d\Upsilon(X_i,X_j), \qquad i,j=1,\ldots,r,
\end{equation}
where $X_1,\ldots,X_r$ is a basis of $V^X$ and $\Upsilon:V^X\times V^X\rightarrow {\rm Cas}(N,L)$ is the map
 (\ref{formula2}), span a finite-dimensional Lie algebra $(E,[[\cdot,\cdot]]_C)$. Indeed,
\[
[[X_i+dT(X_i),X_j+dT(X_j)]]_C=[X_i,X_j]+d\{T(X_i),T(X_j)\}_L, \qquad i,j=1,\ldots,r.
\]
Taking into account that $\{T(X_i),T(X_j)\}_L-T([X_i,X_j])=\Upsilon(X_i,X_j)$, we see that the above is a linear combination of the
generators (\ref{gen}). Additionally, we have that
\[
[[X_i+dT(X_i),d\Upsilon(X_j,X_k)]]_C=d\{T(X_i),\Upsilon(X_j,X_k)\}_L=0.
\]
So, sections (\ref{gen}) span a finite-dimensional subspace $E$ of  $(\Gamma(L),[[\cdot,\cdot]]_C)$. As $\gamma_t\in E$, for all $t\in\mathbb{R}$, we conclude the direct part of the proof.

The converse is straightforward from the fact that $(L,[[\cdot,\cdot]]_C,\rho)$ is a Lie algebroid. Indeed, given the curve $\gamma_t$ within a finite-dimensional real Lie algebra of sections $E$ satisfying that $X_t=\rho(\gamma_t)\in {\rm Ham}(N,L)$, we have that $\{X_t\}_{t\in\mathbb{R}}\subset \rho(E)$ are $L$-Hamiltonian vector fields. As $E$ is a finite-dimensional Lie algebra and $\rho$ is a Lie algebra morphism, $\rho(E)$ is a finite-dimensional Lie algebra of vector fields and $(N,L,X)$ becomes a Dirac--Lie system.
\end{proof}

The above theorem shows the interest of defining a class of Lie systems related to general Lie algebroids.

\section{On diagonal prolongations of Dirac--Lie systems}
The so-called {\it diagonal prolongations} of Lie systems play a fundamental r\^ole in the determination of superposition rules which motivates their study in this section \cite{CGM07}. Specifically, we analyze the properties of diagonal prolongations of Dirac--Lie systems. As a result, we discover new features that can be applied to study their superposition rules and introduce some new concepts of interest.

Let $\tau:E\to N$ be a vector bundle. Its {\it diagonal prolongation} to $N^m$ is the Cartesian product bundle $E^{[m]}=E\times\cdots\times E$ of $m$ copies of $E$, viewed as a vector bundle over $N^m$ in a natural way:
 \begin{equation*}E^{[m]}_{(x_{(1)},\dots,x_{(m)})}=E_{x_{(1)}}\oplus\cdots\oplus E_{x_{(m)}}\,.
\end{equation*}
Every section $X:N\to E$ of $E$ has a natural {\it diagonal prolongation} to a section $X^{[m]}$ of $E^{[m]}$:
\begin{equation*}
X^{[m]}(x_{(1)},\dots,x_{(m)})=X(x_{(1)})+\cdots +X(x_{(m)})\,.
\end{equation*}
Given a function $f:N\rightarrow \mathbb{R}$, we call {\it diagonal prolongation} of $f$ to $N^m$ the function $\widetilde{f}^{[m]}$ on $N^{m}$ of the form
$\widetilde{f}^{[m]}(x_{(1)},\ldots,x_{(m)})= f(x_{(1)})+\ldots+f(x_{(m)})$.

We can consider also sections $X^{(j)}$ of $E^{[m]}$ given by
\begin{equation}\label{prol1}
X^{(j)}(x_{(1)},\dots,x_{(m)})=0+\cdots +X(x_{(j)})+\cdots+0\,.
\end{equation}
It is clear that, if $\{X_i\mid i=1,\ldots,p\}$ is a basis of local sections of $E$, then $\{X_i^{(j)}\mid i=1,\ldots,p,j=1,\ldots,m\}$	 is a basis of local sections of $E^{[m]}$. Note that all this can be repeated also for {\it generalized vector bundles}, like generalized distributions.

Since there are obvious canonical isomorphisms
$$(TN)^{[m]}\simeq TN^m\quad\text{and}\quad (T^*N)^{[m]}\simeq T^*N^m\,,$$
we can interpret the diagonal prolongation $X^{[m]}$ of a vector field on $N$ as a vector field $\widetilde{X}^{[m]}$ on $N^m$, and the diagonal prolongation $\alpha^{[m]}$ of a 1-form on $N$ as a 1-form $\widetilde{\alpha}^{[m]}$ on $N^m$. In the case when $m$ is fixed, we will simply write $\widetilde{X}$ and $\widetilde{\alpha}$. The proof of the following properties of diagonal prolongations is straightforward.

\begin{proposition}\label{p*} The {\it diagonal prolongation} to $N^m$ of a vector field $X$ on
$N$ is the unique vector field $\widetilde X^{[m]}$ on $N^m$, projectable under the map
$\pi:(x_{(1)},\ldots,x_{(m)})\in N^m\mapsto x_{(1)}\in N$ onto $X$ and invariant under the permutation of variables $x_{(i)}\leftrightarrow x_{(j)}$, with $i,j=1,\ldots,m$.
The {\it diagonal prolongation} to $N^m$ of a 1-form $\alpha$ on $N$ is the unique 1-form $\widetilde \alpha^{[m]}$ on $N^m$ such
that $\widetilde \alpha^{[m]}(\widetilde X^{[m]})=\widetilde{\alpha(X)}^{[m]}$ for every vector field $X\in \Gamma(TN)$. We have $d\widetilde\alpha=\widetilde{d\alpha}$ and $\mathcal{L}_{\widetilde X^{[m]}}\widetilde\alpha^{[m]}=\widetilde{\mathcal{L}_X\alpha}^{[m]}$.
In particular, if $\alpha$ is closed (exact), so is its diagonal prolongation $\widetilde \alpha^{[m]}$ to $N^m$.
\end{proposition}

Using local coordinates $(x^a)$ in $N$ and the induced system $(x^a_{(i)})$ of coordinates in $N^m$, we can write, for $X=\sum_aX^a(x)\partial_{x^a}$ and $\alpha=\sum_a\alpha_a(x)dx^a$,
\begin{equation}\label{prol-coord}
\widetilde X^{[m]}=\sum_{a,i}X^a(x_{(i)})\partial_{x^a_{(i)}}\quad\text{and}\quad
\widetilde \alpha^{[m]}=\sum_{a,i}\alpha_a(x_{(i)})dx^a_{(i)}\,.
\end{equation}

Let us fix $m$. Obviously, given two vector fields $X_1$ and $X_2$ on $N$, we have $\widetilde{ [X_1,X_2]}=[\widetilde X_1,\widetilde X_2]$. In consequence, the prolongations to $N^m$ of the elements of a finite-dimensional real Lie algebra $V$ of vector fields on $N$ form a real Lie algebra $\widetilde V$ isomorphic to $V$. Similarly to standard vector fields, we can define the diagonal prolongation of a $t$-dependent vector field $X$ on $N$ to $N^m$ as the only $t$-dependent vector field $\widetilde X$ on $N^m$ satisfying that $\widetilde X_t$ is the prolongation of $X_t$ to $N^m$ for each $t\in\mathbb{R}$.

When $X$ is a Lie--Hamilton system, its diagonal prolongations are also Lie--Hamilton systems in a natural way \cite{BCHLS13}.
Let us now focus on proving an analogue of this result for Dirac--Lie systems.

\begin{definition} Given two Dirac manifolds $(N,L_N)$ and $(M,L_M)$, we say that $\varphi:N\rightarrow M$ is a {\it forward Dirac map} between them if $(L_M)_{\varphi(x)}\!=\!\mathfrak{P}_\varphi(L_N)_x$, where
\[
\mathfrak{P}_\varphi(L_N)_x\!=\!\{\varphi_{*x}X_{x}+\omega_{\varphi(x)}\in T_{\varphi(x)}M\oplus T_{\varphi(x)}^*M\!\mid\! X_x+ (\varphi^*\omega_{\varphi(x)})_x\!\in\! (L_N)_x\},
\]
 for all $x\in N$.
\end{definition}

\begin{proposition} Given a Dirac structure $(N,L)$ and the natural isomorphism
\[
(TN^m\oplus_{N^m}T^*N^m)_{(x_{(1)},\ldots,x_{(m)})}\simeq (T_{x_{(1)}}N\oplus T^*_{x_{(1)}}N)\oplus \cdots\oplus (T_{x_{(m)}}N\oplus T^*_{x_{(m)}}N),
\]
the diagonal prolongation $L^{[m]}$, viewed as a vector subbundle in $TN^m\oplus_{N^m}T^*N^m=\mathcal{P}N^{[m]}$,
is a Dirac structure on $N^m$.

The forward image of $L^{[m]}$ through each $\pi_{i}:(x_{(1)},\ldots,x_{(m)})\in N^m\rightarrow x_{(i)}\in N$, with $i=1,\ldots,m$, equals $L$. Additionally, $L^{[m]}$ is invariant under the permutations $x_{(i)}\leftrightarrow x_{(j)}$, with $i,j=1,\ldots,m$.
\end{proposition}
\begin{proof}
Being a diagonal prolongation of $L$, the subbundle $L^{[m]}$ is invariant under permutations $x_{(i)}\leftrightarrow x_{(j)}$ and each element of a basis $X_i+\alpha_i$ of $L$, with $i=1,\ldots,n$, can naturally be considered as an element $X^{(j)}_i+\alpha^{(j)}_i$ of the $j$th-copy of $L$ within $L^{[m]}$. This gives rise to a basis of $L^{[m]}$, which naturally becomes a smooth $m n$-dimensional subbundle of $\mathcal{P}N^m$. Considering the natural pairing $\langle\cdot,\cdot\rangle_+$ of $\mathcal{P}N^m$ and using $\langle \alpha^{(i)}_j,X^{(k)}_l\rangle=0$ for $i\neq k$, we have
\begin{eqnarray*}
&\left\langle \left(X^{(i)}_{j}+\alpha^{(i)}_{j}\right)(x_{(1)},\ldots,x_{(m)}),\left(X^{(k)}_{l}+
\alpha^{(k)}_{l}\right)(x_{(1)},\ldots,x_{(m)})\right\rangle_+\!\!=\\
&\delta^i_k\left\langle \left(X_{j}+\alpha_{j}\right)(x_{(i)}),\left(X_{l}+\alpha_{l}\right)(x_{(i)})\right\rangle_+=0,
\end{eqnarray*}
for every $p=(x_{(1)},\ldots,x_{(m)})\in N^m$.
As the pairing is bilinear and vanishes on a basis of $L^{[m]}$, it does so on the whole $L^{[m]}$, which is therefore isotropic. Since $L^{[m]}$ has rank $mn$, it is maximally isotropic.

Using that $[X^{(i)}_j,X^{(k)}_l]=0,$ $\iota_{X^{(i)}_j}d\alpha_{l}^{(k)}=0$, and $\mathcal{L}_{X^{(i)}_j}\omega_{l}^{(k)}=0$ for $i\neq k=1,\ldots,m$ and $j,l=1,\ldots,\dim\, N$, we obtain
\begin{equation*}
[[X^{(i)}_j+\alpha^{(i)}_j,X^{(k)}_l+\alpha^{(k)}_l]]_C=\delta^{i}_{k}[[X^{(i)}_j+\alpha^{(i)}_j,X^{(i)}_l+\alpha^{(i)}_l]]_C\in \Gamma(L^{[m]}).
\end{equation*}
So, $L^{[m]}$ is involutive. Since it is also maximally isotropic, it is a Dirac structure.

Let us prove that $\mathfrak{P}_{\pi_a}(L^{[m]})=L$ for every $\pi_a$. Note that
$(X^{(a)}_j+\alpha^{(a)}_j)_p\in L^{[m]}_p$ is such that $\pi_{a*}(X^{(a)}_j)_p=(X_j)_{x_{(a)}}$ and $(\alpha_j)_{x_{(a)}}\circ(\pi_{*a})_p =(\alpha^{(a)}_j)_p$ for every $p\in \pi^{-1}_a(x_{(a)})$. So, $(X_j+\alpha_j)_{x_{(a)}}\in (\mathfrak{P}_{\pi_a}(L^{[m]}))_{x_{(a)}}\subset L_{x_{(a)}}$ for $j=1,\ldots,n$ and every $x_{(a)}\in N$. Using that $X_j+\alpha_j$ is a basis for $L$ and the previous results, we obtain $L\subset \mathfrak{P}_{\pi_a}(L^{[m]})$. Conversely, $\mathfrak{P}_{\pi_a}(L^{[m]})\subset L$. Indeed, if $(X+\alpha)_{x_{(a)}}\in \mathfrak{P}_a(L^{[m]})$, then there exists an element $(Y+\beta)_p\in L^{[m]}_p$, with $p\in \pi^{-1}(x_{(a)})$, such that $\pi_{a*}Y_p=X_{x_{(a)}}$ and $(\alpha)_{x_{(a)}}\circ (\pi_{*a})_p=\beta_p$. Using that $(Y+\beta)_p=\sum_{ij}c_{ij}(X^{(i)}_j+\alpha^{(i)}_j)_p$ for a unique set of constants $c_{ij}$, with $i=1,\ldots,m$ and $j=1,\ldots,n$, we have
$\pi_{a*}(\sum_{ij}c_{ij}(X^{(i)}_j)_p)=\sum_{j}c_{aj}(X_j)_{x_{(a)}}=X_{x_{(a)}}$. Meanwhile, $\beta_p=\alpha_{x_{(a)}}\circ (\pi_{*a})_p$ means that
$\sum_{j}c_{aj}(\alpha_j)_{x_{(a)}}=\alpha_{x_{(a)}}$. So, $(X+\alpha)_{x_{(a)}}=\sum_{j}c_{aj}(X_j+\alpha_j)_{x_{(a)}}\in L_{x_{(a)}}$.
\end{proof}

\begin{corollary}\label{COR1} Given a Dirac structure $(N,L)$, we have $\rho_m (L^{[m]})=\rho(L)^{[m]}$, where $\rho_m$ is the projection $\rho_m:\mathcal{P}N^m\rightarrow TN^m$.  Then, if $X$ is an $L$-Hamiltonian vector field
with respect to $L$, its diagonal prolongation $\widetilde X^{[m]}$ to $N^m$ is an $L$-Hamiltonian vector field with respect to $L^{[m]}$. Moreover, $\rho^*_m (L^{[m]})=\rho^*(L)^{[m]}$, where $\rho^*_m$ is the canonical projection $\rho^*_m:\mathcal{P}N^m\rightarrow T^*N^m$.
\end{corollary}

\begin{corollary} If $(N,L,X)$ is a Dirac--Lie system, then
$(N^m,L^{[m]},\widetilde X^{[m]})$ is also a Dirac--Lie system.
\end{corollary}
\begin{proof} If $X$ admits a Vessiot--Guldberg Lie algebra $V$ of Hamiltonian vector fields with respect to $(N,L)$, then $\widetilde X$ possesses a Vessiot--Guldberg Lie algebra $\widetilde V$ given by the diagonal prolongations of the elements of $V$, which are $L^{[m]}$-Hamiltonian vector fields, by construction of $L^{[m]}$ and Corollary \ref{COR1}.
\end{proof}

Similarly to the prolongations of vector fields, one can define prolongations of functions and 1-forms in an obvious way.

\begin{proposition}\label{Prop} Let $X$ be a vector field and $f$ be a function on $N$. Then:
\begin{itemize}
\item[(a)] If $f$ is an $L$-Hamiltonian function for $X$, its diagonal prolongation $\widetilde{f}$ to $N^m$ is an $L^{[m]}$-Hamiltonian function of the diagonal prolongation $\widetilde X$ to $N^m$.
\item[(b)] If $f\in {\rm Cas}(N,L)$, then $\widetilde f\in {\rm Cas}(N^m,L^{[m]})$.
\item[(c)] The map $\lambda:({\rm Adm}(N,L),\{\cdot,\cdot\}_L) \ni f \mapsto \widetilde f \in ({\rm Adm}(N^m,L^{[m]}),\{\cdot,\cdot\}_{L^{[m]}})$ is an injective Lie algebra morphism.

\end{itemize}
\end{proposition}
\begin{proof}
Let $f$ be an $L$-Hamiltonian function for $X$. Then, $X+df\in \Gamma(L)$ and $\widetilde X+d\widetilde f=\widetilde X+\widetilde{df}$ is as an element of $\Gamma(L^{[m]})$. By a similar argument, if $f\in {\rm Cas}(N,L)$, then $\widetilde f\in {\rm Cas}(N^m,L^{[m]})$.
Given $f,g\in {\rm Adm}(N,L)$, we have $\widetilde{\{f,g\}_L}=\widetilde{X_fg}=\widetilde{X_f}\widetilde{g}=X_{\widetilde{f}}\widetilde{g}=\{\widetilde f,\widetilde g\}_{L^{[m]}}$, i.e., $\lambda(\{f,g\}_L)=\{\lambda(f),\lambda(g)\}_{L^{[m]}}$. Additionally, as $\lambda$ is
linear, it becomes a Lie algebra morphism. Moreover, it is easy to see that $\widetilde f=0$ if and only if $f=0$. Hence, $\lambda$ is injective.
\end{proof}
Note, however, that in the above we cannot ensure that $\lambda$ is a Poisson algebra morphism, as in general $\widetilde {fg}\neq \widetilde f\widetilde g$.

Using the above proposition, we can easily prove the following corollaries.
\begin{corollary} If $h_1,\ldots,h_r:N\rightarrow \mathbb{R}$ is a family of functions on a Dirac manifold $(N,L)$ spanning a finite-dimensional real Lie algebra of functions with respect to the Lie bracket $\{\cdot,\cdot\}_L$, then their diagonal prolongations $\widetilde h_1,\ldots,\widetilde h_r$ to $N^m$ close an isomorphic Lie algebra of functions with respect to the Lie bracket $\{\cdot,\cdot\}_{L^{[m]}}$ induced by the Dirac structure $(N^m,L^{[m]})$.
\end{corollary}

\begin{corollary} If $(N,L,X)$ is  a Dirac--Lie system admitting a Lie--Hamiltonian  $(N,L,h)$, then $(N^m,L^{[m]},\widetilde X^{[m]})$ is a Dirac--Lie system with a Dirac--Lie Hamiltonian $(N^m,L^{[m]},h^{[m]})$, where $h^{[m]}_t=\widetilde h_t^{[m]}$  is the diagonal prolongation of $h_t$ to $N^m$.
\end{corollary}

\section{Superposition rules and $t$-independent constants of the motion for
Dirac--Lie systems}

Let us give a first straightforward application of Dirac--Lie systems to obtain constants of the motion.

\begin{proposition}
  \label{Cas}Given a Dirac--Lie system $(N, L, X)$, the elements of
  ${\rm Cas}(N, L)$ are constants of the motion for $X$. Moreover, the set $\mathcal{I}^X_L$ of
 its admissible $t$-independent constants of the motion form a Poisson
  algebra $( \mathcal{I}^X_L, \sbt, \{ \cdot, \cdot \}_L)$.
\end{proposition}

\begin{proof}
  Two admissible functions $f$ and $g$ are $t$-independent constants of the motion for $X$ if and
  only if $X_t f = X_t g = 0$ for every $t \in \mathbb{R}$. Using that every
$X_t$ is a derivation of the associative algebra  $C^\infty(N)$, we see that given $f,g\in \mathcal{I}^X_L$, then $f+g$, $\lambda f$,  and $f\cdot g$ are
also constants of the motion for $X$ for every $\lambda\in\mathbb{R}$. Since the sum and product of admissible functions are admissible functions, then $\mathcal{I}_L^X$ is closed under the sum and product of elements and real constants. So $(\mathcal{I}_L^X,\sbt)$ is an associative
subalgebra of $(C^\infty(N),\sbt)$.

As $(N,L,X)$ is a Dirac--Lie system, the vector fields $\{X_t \}_{t \in
\mathbb{R}}$
  are $L$-Hamiltonian. Therefore,
  \[ X_t \{f, g\}_L = \{X_tf, g\}_L + \{f, X_tg \}_L.
 \]
 As $f$ and $g$ are constants
  of the motion for $X$, then $\{f, g\}_L$ is so also. Using that $\{f,g\}_L$ is also an admissible function, we
finish the proof.
\end{proof}
The following can easily be proved.
\begin{proposition}\label{Prop:Ham}
  Let $(N,L,X)$ be a Dirac--Lie system possessing a Dirac--Lie Hamiltonian
  $(N, L, h)$. An admissible function $f : N \rightarrow \mathbb{R}$ is
  a constant of the motion for $X$ if and only if it Poisson commutes with all
  the elements of Lie$(\{h_t \}_{t \in \mathbb{R}}, \{ \cdot, \cdot
  \}_{L})$.
\end{proposition}

Consider a Dirac--Lie system $(N,L^\omega,X)$ with $\omega$ being
a symplectic structure and $X$ being an autonomous system. Consequently, ${\rm Adm}(N,L)=C^\infty(N)$ and the above
proposition entails that $f\in C^\infty(N)$ is a constant of the motion for $X$ if and only if it Poisson commutes
with a Hamiltonian function $h$ associated to $X$. This shows that Proposition \ref{Prop:Ham}
recovers as a particular case this well-known result \cite{FM}.
Additionally, Proposition \ref{Prop:Ham} suggests us that the r\^ole played by
autonomous Hamiltonians for autonomous Hamiltonian systems is performed by
finite-dimensional Lie algebras of admissible functions associated with a
Dirac--Lie
Hamiltonian for Dirac--Lie systems. This fact can be employed, for instance, to study $t$-independent
first-integrals of Dirac--Lie systems, e.g., the maximal
number of such first-integrals in involution, which would lead to the
interesting analysis of integrability/superintegrability and action/angle variables for Dirac--Lie
systems \cite{NTZ12}.

Another reason to study $t$-independent constants of the motion of
Lie systems is their
use in deriving superposition rules \cite{CGM00}. More
explicitly, a superposition rule for a Lie system can be obtained through the $t$-independent
constants of the motion of one of its diagonal
prolongations \cite{CGM07}. The following proposition provides some ways of obtaining such constants.

\begin{proposition} If $X$ be a system possessing a $t$-independent constant of the motion $f$, then:
\begin{enumerate}
 \item The diagonal prolongation $\widetilde f^{[m]}$ is a $t$-independent constant of the motion for $\widetilde X^{[m]}$.
 \item If $Y$ is a $t$-independent Lie symmetry of $X$, then $\widetilde Y^{[m]}$ is a $t$-independent Lie symmetry of $\widetilde X^{[m]}$.
 \item If
$h$ is a $t$-independent constant of the motion for $\widetilde X^{[m]}$, then $\widetilde Y^{[m]}h$ is another $t$-independent constant of the motion for $\widetilde X^{[m]}$.
\end{enumerate}
\end{proposition}
\begin{proof} This result is a straightforward application of Proposition \ref{p*} and the properties of the diagonal prolongations of $t$-dependent vector fields.
\end{proof}

Using the fact that the diagonal prolongation of vector fields is a Lie bracket homomorphism, in virtue of Proposition \ref{invariants} we get the following.
\begin{proposition}\label{Prop:Cas} Given a Dirac--Lie system $(N,L,X)$ that admits a Dirac--Lie Hamiltonian  $(N,L,h)$ such that $\{h_t\}_{t\in\mathbb{R}}$ is contained in a finite-dimensional  Lie algebra of admissible functions $(\mathfrak{M},\{\cdot,\cdot\}_L)$. Given the momentum map $J:N^m	\to\mathfrak{W}^*$ associated with the Lie algebra morphism $\phi:f\in \mathfrak{W}\mapsto \widetilde f\in {\rm Adm}(N^m,L^{[m]})$, the pull-back $J^*(C)$  of any Casimir function $C$ on $\mathfrak{W}^*$ is a constant of the motion for the diagonal prolongation $\widetilde X^{[m]}$. If $\mathfrak{W}\simeq{\rm Lie}(\{\widetilde h_t\}_{t\in\mathbb{R}},\{\cdot,\cdot\}_{L^{[m]}})$, the function $J^*(C)$ Poisson commutes with all $L^{[m]}$-admissible constants of the motion of $\widetilde X^{[m]}$.
\end{proposition}
\subsection{Example}
Let us use the above results to devise a superposition rule for the third-order Kummer--Schwarz equation in first-order form (\ref{firstKS3}) with $c_0=0$, the so-called {\it Schwarzian equations} \cite{SS12,TTX01}. To simplify the presentation, we will always assume $c_0=0$ in this section. It is known (cf. \cite{CGL11}) that the derivation of a superposition rule for this system can be reduced to obtaining certain three $t$-independent constants of the motion for the diagonal prolongation $\widetilde{X}^{3KS}$ of $X^{3KS}$ to $\mathcal{O}_2^2$. In \cite{CGL11} such
constants were worked out through the method of characteristics which consists in solving a series of systems of ODEs. Nevertheless, we can determine such constants more easily through Dirac--Lie systems.

The $t$-dependent vector field $\widetilde X^{3KS}$ is spanned by a linear combination of the diagonal prolongations of $Y_1$, $Y_2,$ and $Y_3$ to $\mathcal{O}_2^2$. From (\ref{VFKS1}), we have
\[
\begin{gathered}
\widetilde Y_1=\sum_{i=1}^2v_i\frac{\partial}{\partial a_i},\quad \widetilde Y_2=
\sum_{i=1}^2\left(v_i\frac{\partial}{\partial v_i}+2a_i\frac{\partial}{\partial a_i}\right),\\ \widetilde Y_3=\sum_{i=1}^2\left(v_i\frac{\partial}{\partial x_i}+a_i\frac{\partial}{\partial v_i}+\frac 32
\frac{a_i^2}{v_i}\frac{\partial}{\partial a_i}\right).
\end{gathered}
\]
From Proposition \ref{Prop} and functions (\ref{Fun3KS}), the vector fields $\widetilde Y_1,\widetilde Y_2,\widetilde Y_3$ are $L^{[2]}$-Hamiltonian with $L^{[2]}$-Hamiltonian functions
\[
\widetilde h_1=-\frac 2{v_1}-\frac 2{v_2},\qquad \widetilde {h}_2=-\frac{a_1}{v_1^2}-\frac{a_2}{v_2^2},\qquad \widetilde {h}_3=-\frac {a_1^2}{2v_1^3}-\frac {a_2^2}{2v_2^3}.
\]
Indeed, these are the diagonal prolongations to $\mathcal{O}_2^2$ of the $L$-Hamiltonian functions of $Y_1, Y_2$, and $Y_3$.
 Moreover, they span a real Lie algebra of functions isomorphic to that one spanned by $h_1,h_2,h_3$ and to $\mathfrak{sl}(2,\mathbb{R})$. We can then define a Lie algebra morphism $\phi:\mathfrak{sl}(2,\mathbb{R})\rightarrow C^\infty(N^2)$ of the form $\phi(e_1)=\widetilde h_1$, $\phi(e_2)=\widetilde h_2$ and $\phi(e_3)=\widetilde h_3$, where $\{ e_1,e_2,e_3\}$ is the standard basis of $\mathfrak{sl}(2,\mathbb{R})$. Using that $\mathfrak{sl}(2,\mathbb{R})$ is a simple Lie algebra, we can compute the Casimir invariant on $\mathfrak{sl}(2,\mathbb{R})^*$ as $e_1e_3-e_2^2$ (where $e_1,e_2,e_3$ can be considered as functions on $\mathfrak{sl}(2,\mathbb{R})$). Proposition 	\ref{Prop:Cas} ensures then that $\widetilde h_1\widetilde h_3-\widetilde h_2^2$ Poisson commutes with $\widetilde h_1,\widetilde h_2$ and $\widetilde h_3$. In this way, we obtain a constant of the motion for $\widetilde{X}^{3KS}$ given by
 \[
 I=\widetilde h_1\widetilde h_3-\widetilde h_2^2=\frac{(a_2v_1-a_1v_2)^2}{v_1^3v_2^3}.
 \]
Schwarzian equations admit a Lie symmetry $Z=x^2\partial/ \partial x$ \cite{OT09}. Its prolongation to ${T}^2\mathbb{R}$, i.e.,
\begin{equation}\label{LieSym}
Z_P=x^2\frac{\partial}{\partial x}+2vx\frac{\partial}{\partial v}+2(ax+v^2)\frac{\partial}{\partial a},
\end{equation}
is a Lie symmetry of $X^{3KS}$. From Proposition \ref{Prop}, we get that $\widetilde Z_P$ is a Lie symmetry of $\widetilde{X}^{3KS}$. So, we can construct constants of the motion for $\widetilde{X}^{3KS}$ by applying $\widetilde Z_P$ to any of its $t$-independent constants of the motion. In particular,
\[
F_2\equiv -\widetilde Z_P \log |I|=x_1+x_2+\frac{2v_1v_2(v_1-v_2)}{a_2v_1-a_1v_2}
\]
is constant on particular solutions $(x_{(1)}(t),v_{(1)}(t),a_{(1)}(t),x_{(2)}(t),v_{(2)}(t),a_{(2)}(t))$ of $\widetilde{X}^{3KS}$. If $(x_{(2)}(t),v_{(2)}(t),a_{(2)}(t))$ is a particular solution for $X^{3KS}$, its opposite is also. So, the function
\[
F_3\equiv x_1-x_2+\frac{2v_1v_2(v_1+v_2)}{a_2v_1-a_1v_2}
\]
is also constant along solutions of $\widetilde{X}^{3KS}$, i.e., it is a new constant of the motion. In consequence, we get three $t$-independent constants of the motion: $\Upsilon_1=I$ and
\[
\Upsilon_2=\frac{F_2+F_3}2=x_1+\frac{2v_1^2v_2}{a_2v_1-a_1v_2},\qquad \Upsilon_3=\frac{F_2-F_3}2=x_2-\frac{2v_1v_2^2}{a_2v_1-a_1v_2}.
\]
This gives rise to three $t$-independent constants of the motion for ${\widetilde X}^{3KS}$. Taking into account that $\partial (\Upsilon_1,\Upsilon_2,\Upsilon_3)/\partial (x_1,v_1,a_1)\neq 0$, the expressions $\Upsilon_1=\lambda_1$, $\Upsilon_2=\lambda_2$, and $\Upsilon_3=\lambda_3$ allow us to obtain the expressions of $x_1,v_1,a_1$ in terms of the remaining variables and $\lambda_1,\lambda_2,\lambda_3$. More specifically,
\[
x_1=\frac{4}{\lambda_1(\lambda_3-x_2)}+\lambda_2,\,\, v_1=\frac{4v_2}{\lambda_1(\lambda_3-x_2)^2},\,\, a_1=\frac{8v_2^2+4 a_2(\lambda_3-x_2)}{\lambda_1(\lambda_3-x_2)^3}.
\]
According to the theory of Lie systems \cite{CGM07}, the map $\Phi:(x_2,v_2,a_2;\lambda_1,\lambda_2,\lambda_3)\in \mathcal{O}_2^2\times \mathbb{R}^3\mapsto (x_1,v_1,a_1)\in \mathcal{O}_2^2$ enables us to write the general solution of (\ref{firstKS3}) into the form
\[
(x(t),v(t),a(t))=\Phi(x_2(t),v_2(t),a_2(t);\lambda_1,\lambda_2,\lambda_3).
\]
This is the known superposition rule for Schwarzian equations (in first-order form) derived in \cite{LS12} by solving a system of PDEs. Meanwhile, our present techniques enable us to obtain the same result without any integration. Note that $x(t)$, the general solution of Schwarzian equations, can be written as $x(t)=\tau\circ \Phi(x_2(t),\lambda_1,\lambda_2,\lambda_3)$, with $\tau$ the projection $\tau:(x_2,v_2,a_2)\in { T}^2\mathbb{R}\mapsto x_2\in \mathbb{R}$, from a unique particular solution of (\ref{KS3}), recovering a known feature of these equations \cite{OT09}.

\section{Bi--Dirac--Lie systems}

It can happen that a Lie system $X$ on a manifold $N$ possesses Vessiot--Guldberg Lie algebras of vector fields with respect to two different Dirac structures. This results in defining two Dirac--Lie systems. For instance, the system of coupled Riccati equations (\ref{CoupledRic}) admits two Dirac--Lie
structures \cite{BHLS12Cla}: the one previously given, $(\mathcal{O},L^\omega,X)$, where $\omega$ is given by (\ref{SymCoupledRic}), and a second one, $(\mathcal{O},L^{\bar \omega},X)$, with
\[
\bar \omega=\sum_{i< j=1}^4\frac{dx_i\wedge x_j}{(x_i-x_j)^2}.
\]
In the following sections, several similar examples will be shown. This suggests us to define the following notion.

\begin{definition} A {\it bi--Dirac--Lie system} is a four-tuple $(N,L_1,L_2,X)$, where $(N,L_1)$ and $(N,L_2)$ are two different Dirac manifolds and $X$ is a Lie system on $N$ such that $V^X\subset {\rm Ham}(N,L_1)\cap {\rm Ham}(N,L_2)$.
\end{definition}

Given a bi--Dirac--Lie system $(N,L_1,L_2,X)$, we can apply indistinctly the methods of the previous sections to $(N,L_1,X)$ and $(N,L_2,X)$ to obtain superposition rules, constants of the motion, and other properties of $X$. This motivates studies on constructions of this type of structures.

Let us depict a new procedure to build up bi--Dirac--Lie systems from $(N,L^\omega,X)$ whose $X$ possesses a $t$-independent Lie symmetry $Z$.  This method is a generalization to nonautonomous systems, associated to presymplectic manifolds, of the method  for autonomous Hamiltonian systems devised in \cite{CMR02}.

Consider a Dirac--Lie system $(N,L^{\omega},X)$, where $\omega$ is a presymplectic structure, and a $t$-independent Lie symmetry $Z$  of $X$, i.e. $[Z, X_t] = 0$ for all $t \in
\mathbb{R}$. Under the  above assumptions, $\omega_Z=\mathcal{L}_Z\omega$ satisfies $d\omega_Z=d \mathcal{L}_Z \omega =
\mathcal{L}_Z d
\omega = 0$, so $(N,\omega_Z)$ is a presymplectic manifold. The vector fields of $V^X$ are still
Hamiltonian with respect to $\left( N, \omega_Z \right)$. Indeed, we can see that Theorem \ref{MT} ensures that $X$
admits a Dirac--Lie Hamiltonian $(N,L^\omega,h)$ and
\[
\begin{gathered}
\left[Z,X_t\right]=0\Longrightarrow \iota_{X_t}\circ\mathcal{L}_Z=\mathcal{L}_Z\circ\iota_{X_t}\Longrightarrow \\\iota_{X_t} \omega_Z = \iota_{X_t}  \mathcal{L}_Z
  \omega = \mathcal{L}_Z \iota_{X_t} \omega = -\mathcal{L}_Z dh_t =- d
(Zh_t),\qquad \forall t\in\mathbb{R}.
\end{gathered}
\]
So, the vector fields $\{X_t\}_{t\in\mathbb{R}}$ are
$L^{\omega_Z}$-Hamiltonian. Since the successive Lie brackets and linear combinations of
$L$-Hamiltonian vector fields and elements of $V^X$ are $L$-Hamiltonian vector fields, the whole Lie algebra  $V^X$ is Hamiltonian with respect to $\omega_Z$. Consequently, $(N,L^{\omega_Z},X)$ is also a Dirac--Lie system. In view of (\ref{PreSymSeqII}) and since $\dim {\rm Cas}(N,L^{\omega_Z})=1$, we see that $(B^{\omega_Z})^{-1}(V^X)$ is a finite-dimensional Lie algebra. As the curve $\bar h:t\in\mathbb{R}\mapsto Zh_t\in {\rm Adm}(N,L^{\omega_Z})$ is included within $(B^{\omega_Z})^{-1}(V^X)$,  the Lie algebra ${\rm Lie}(\{Zh_t\}_{t\in\mathbb{R}},\{\cdot,\cdot\}_{L^{\omega_Z}})$, where $\{\cdot,\cdot\}_{L^{\omega_Z}}$ is the Poisson bracket induced by $L^{\omega_Z}$, becomes finite-dimensional. In other words, $(N,L^{\omega_Z},Zh_t)$ is also a Lie--Hamiltonian for $X$. Moreover,
\[\{\bar h_t,\bar h_{t'}\}_{L^{\omega_Z}}=X_t(\bar h_{t'})=X_t(Zh_{t'})=Z(X_th_{t'})=Z\{ h_t,h_{t'}\}_{L^{\omega}},\qquad \forall t\in\mathbb{R}\,.
\]Summarizing, we have the following proposition.

\begin{proposition} If $(N,L^\omega,X)$ is a Dirac--Lie system for which $X$ admits a $t$-independent Lie symmetry $Z$,
then $(N,L^\omega,L^{\mathcal{L}_Z\omega},X)$ is a bi--Dirac--Lie system. If $(N,L^\omega,h)$ is a Dirac--Lie Hamiltonian for $X$, then $(N,L^{\mathcal{L}_Z\omega},Zh)$ is a Dirac--Lie Hamiltonian for $X$ and there exists an exact sequence of Lie algebras
$$
(\{h_t\}_{t\in\mathbb{R}},\{\cdot,\cdot\}_{L^{\omega}})\stackrel{Z}{\longrightarrow} (\{Zh_t\}_{t\in\mathbb{R}},\{\cdot,\cdot\}_{L^{\omega_Z}})\rightarrow 0\,.
$$
\end{proposition}

Note that, given a Lie--Hamilton system $(N,L^\omega,X)$, the triple $(N,L^{\omega_Z},X)$ need not be a Lie--Hamilton system: $\omega_Z$ may fail to be a symplectic two-form (cf. \cite{CMR02}). This causes that the theory of Lie--Hamilton systems cannot be applied to study $(N,L^{\omega_Z},X)$, while the  methods of our work do.
\subsection{Example}
Let us illustrate the above theory with an example. Recall that Schwarzian equations  admit a Lie symmetry $Z=x^2\partial/\partial x$. As a consequence, system (\ref{firstKS3}), with $c_0=0$, possesses a $t$-independent Lie symmetry $Z_P$ given by (\ref{LieSym}) and
\[
\omega_{Z_P}\equiv \mathcal{L}_{Z_P}\omega_{3KS}=-\frac{2}{v^3}(xdv\wedge da+vda\wedge dx+adx\wedge dv).
\]
Moreover,
\[
\begin{gathered}
\iota_{Y_1}\omega_{Z_P}=-d({{Z_P}}h_1)=-d\left(\frac{4x}{v}\right),\quad
\iota_{Y_2}\omega_{Z_P}=-d({{Z_P}}h_2)=d\left(2-\frac{2ax}{v^2}\right),\\
\iota_{Y_3}\omega_{Z_P}=-d({{Z_P}}h_3)=d\left(\frac {2a}v-\frac{a^2x}{v^3}\right).
\end{gathered}
\]
So, $Y_1,Y_2$, and $Y_3$ are Hamiltonian vector fields with respect to $\omega_{Z_P}$. Moreover, since
\begin{eqnarray*}\{ {Z_P}h_1, {Z_P}h_2\}_{L^{\omega_{Z_P}}}&=& Z_Ph_1\,,\\
\{ {Z_P}h_2, {Z_P}h_3\}_{L^{\omega_{Z_P}}}&=& {Z_P}h_3\,,\\
\{ {Z_P}h_1,{Z_P}h_3\}_{L^{\omega_{Z_P}}}&=&2{Z_P}h_2\,,
\end{eqnarray*}
we see that $Z_P h_1$, $Z_Ph_2,$ and $Z_Ph_3$ span a new finite-dimensional real Lie algebra. So,
if $(\mathcal{O}_2,L^\omega,h)$ is a Lie--Hamiltonian for $X$, then $(\mathcal{O}_2,L^{\omega_{Z_P}},Z_Ph)$ is a Dirac--Lie Hamiltonian for $X$.

Let us devise a more general  method to construct bi--Dirac--Lie systems. Given a Dirac manifold $(N,L)$ and a closed two-form $\omega$ on $N$, the sections on $TN\oplus_N T^*N$ of the form
\[
X+\alpha-\iota_X\omega,
\]
where $X+\alpha\in \Gamma(L)$, span a new Dirac structure $(N,\,^\omega\! L)$ \cite{BR03}. When two Dirac structures are connected by a transformation of this type, it is said that they are {\it gauge equivalent}. Using this, we can prove the following propositions.
\begin{proposition} Let $Z$ be a vector field on $N$. Then, the Dirac structures $L^\omega$ and $L^{\omega_Z}$, with $\omega_Z=\mathcal{L}_Z\omega$, are gauge equivalent.
\end{proposition}
\begin{proof} The Dirac structure $L^\omega$ is spanned by sections of the form $X-\iota_X\omega$, with $X\in \Gamma(N)$, and the Dirac structure $L^{\omega_Z}$ is spanned by sections of the form $X-\iota_X\omega_Z$. Recall that $d\omega=d\omega_Z=0$. So, $L^{\omega_Z}$ is of the form
\[
X-\iota_X\omega-\iota_X(\omega_Z-\omega),\qquad X-\iota_X\omega\in \Gamma(L^\omega).
\]
As $d(\omega_Z-\omega)=0$, then $L^\omega$ and $L^{\omega_Z}$ are connected by a gauge transformation.
\end{proof}
This result gives us a hint to construct a more general method to create bi--Dirac--Lie systems.

\begin{proposition} Let $(N,L,X)$ be a Dirac--Lie system and $\omega$ be a closed two-form such that $\widehat \omega(V^X) \subset \mathcal{B}^1(N)$, where $\mathcal{B}^1(N)$ is the space of exact one-forms on $N$. Then,
$(N,L,^\omega\!\! L,X)$ is a bi--Dirac--Lie system.
\end{proposition}

\begin{proof} If $Y\in V^X$, then it is $L$-Hamiltonian and $Y+df\in \Gamma(L)$ for a certain function $f\in C^\infty(N)$. By definition of $\,^\omega\! L$, we have that $Y+df-\iota_Y\omega \in \Gamma(\,^\omega\! L)$. Since $\widehat\omega(V^X)\subset \mathcal{B}^1(N)$ by assumption, then $\widehat \omega(Y)=-dg$ for a certain $g\in C^\infty(N)$.  So,
$Y+d(f+g)\in \Gamma(\,^\omega\! L)$ and $Y$ is $\,^\omega\! L$-Hamiltonian. Hence, $V^X$ is a finite-dimensional real Lie algebra of $\,^\omega\! L$-Hamiltonian vector fields, $(N,\,^\omega\! L,X)$ is a Dirac--Lie system and $(N,L,\,^\omega\! L,X)$ is a bi--Dirac--Lie system.
\end{proof}

\begin{note} Note that two gauge equivalent Dirac structures may have different spaces of admissible functions. This causes that they can be used to obtain different admissible constants of the motion and other properties of $X$. In brief, gauge equivalent Dirac structures are not equivalent from the point of view of their associated Dirac--Lie systems.
\end{note}

\section{Dirac--Lie systems and mixed superposition rules}

In this section we will use the developed methods of Dirac--Lie systems to constructing mixed superposition rules.

Recall that a {\it mixed superposition rule} for a system $X$ on $\mathbb{R}^n$, in terms of some systems $X_{(1)},\ldots,X_{(m)}$, is a superposition function
$\Phi:\mathbb{R}^{n_1}\times\ldots\mathbb{R}^{n_m}\times\mathbb{R}^{n}\rightarrow\mathbb{R}^n$ allowing us to express the general solution, $x(t)$, of $X$ in the form
\[
x(t)=\Phi(x_{(1)}(t),\ldots,x_{(m)}(t),\lambda_1,\ldots,\lambda_n),
\]
where $\lambda_1,\ldots,\lambda_n$ are real constants and $x_{(1)}(t),\ldots,x_{(m)}(t)$ are particular solutions of the systems $X_{(1)},\ldots,X_{(m)}$, respectively.
The main advantage of the use of mixed superposition rules is that they are much more versatile than standard superposition rules \cite{GL12}.

In \cite{GL12} it was proved that a mixed superposition rule for a Lie system $X$ on $\mathbb{R}^n$ can be obtained by the following procedure.
We have to determine a series of systems
\[
X_{(a)}=\sum_{i=1}^{n_a}X^i_{(a)}(t,x_{(a)})\frac{\partial}{\partial x^i_{(a)}}\in \Gamma({T}\mathbb{R}^{n_a}),\qquad a=1,\ldots,m,
\]
such that $X_E=X_{(1)}\times\ldots\times X_{(m)}\times X$, i.e.,
the time-dependent vector field
\[
X_E(t,x_{(1)},\ldots,x_{(m)})=\sum_{i=1}^nX^i(t,x)\frac{\partial}{\partial x^i}+\sum_{a=1}^m \sum_{i=1}^{n_a}X^i_{(a)}(t,x_{(a)})\frac{\partial}{\partial x^i_{(a)}},
\]
gives rise to the distribution $\mathcal{D}^{X_E}$ for which the projection
\[
{\rm pr}_*:\mathcal{D}^{X_E}\rightarrow {T}(\mathbb{R}^{n_1}\times\ldots\times \mathbb{R}^{n_m})\,,\quad\text{with}\quad {\rm pr}(x_{(1)},\ldots,x_{(m)},x)= (x_{(1)},\ldots,x_{(m)})\,,
\]
is an injective map. In such a case, a family $F_1,\ldots,F_n:\mathbb{R}^{n_1}\times\ldots\times\mathbb{R}^{n_m}\times\mathbb{R}^n\rightarrow\mathbb{R}$ of  $t$-independent constants of the motion for $X_E$
satisfying
\[
\frac{\partial(F_1,\ldots,F_n)}{\partial(x_1,\ldots,x_n)}\neq 0,
\]
where $(x_1,\ldots,x_n)$ is the coordinate system on $\mathbb{R}^n$, enables us to construct a mixed superposition rule.
Indeed, the equations $F_i=\lambda_i$, where $\lambda_1,\ldots,\lambda_n$ are real constants, allow us to obtain the variables $x_1,\ldots,x_n$ in terms of the remaining variables $x_{(1)},\ldots,x_{(m)}$ and $\lambda_1,\ldots,\lambda_n$, giving rise to a map
\[
(x_1,\ldots,x_n)=\Phi(x_{(1)},\ldots,x_{(m)};\lambda_1,\ldots,\lambda_n),
\]
which becomes, along with $X_{(1)},\ldots,X_{(m)}$, the searched mixed superposition rule for $X$.
\subsection{Example}
We now aim to obtain a mixed superposition rule to study the Schwarzian equation
\begin{equation}\label{Schwarz}
\{x,t\}=\frac{d^3x}{dt^3}\left(\frac{dx}{dt}\right)^{-1}-\frac 32\left(\frac{d^2x}{dt^2}\right)^2\left(\frac{dx}{dt}\right)^{-2}=2b_1(t),
\end{equation}
where $\{x,t\}$ is the referred to as {\it Schwarzian derivative} of the function $x(t)$ in terms of the variable $t$ and $b_1(t)$ is an arbitrary nonconstant $t$-dependent function. More specifically, we obtain a mixed superposition rule for the Lie system (\ref{firstKS3}) with $c_0=0$, which is obtained from Schwarzian equations by adding two variables $v=dx/dt$ and $a=dv/dt$. Then, we use the mixed
superposition rule to analyze (\ref{Schwarz}).

In order to determine the searched mixed superposition rule, consider for example the direct product of (\ref{Schwarz}) along with the Lie systems
\begin{equation}\label{linear}
\left\{\begin{aligned}
\frac{dx_{(i)}}{dt}&=v_{(i)},\\
\frac{dv_{(i)}}{dt}&=-b_1(t)x_{(i)},\\
\end{aligned}\right.\qquad i=1,2.
\end{equation}
The above systems can be written in the form $(X_t)_{(i)}=X^3_{(i)}+b_1(t)X^1_{(i)}$, with $i=1,2$ and
\[
X^1_{(i)}=-x_{(i)}\frac{\partial}{\partial v_{(i)}},\qquad X^2_{(i)}=\frac 12\left(v_{(i)}\frac{\partial}{\partial v_{(i)}}-x_{(i)}\frac{\partial}{\partial x_{(i)}}\right), \qquad X^3_{(i)}=v_{(i)}\frac{\partial}{\partial x_{(i)}}.
\]
Since $X^1_{(i)}$, $X^2_{(i)},$ and $X^3_{(i)}$, with $i=1,2$, close the same commutation relations as the vector fields $Y_1,Y_2,$ and $Y_3$ given by (\ref{VFKS1}), we obtain that the vector fields
\[
M^1\equiv X^1_{(1)}\times X^1_{(2)}\times Y_1,\qquad M^2\equiv X^2_{(1)}\times X^2_{(2)}\times Y_2,\qquad M^3\equiv X^3_{(1)}\times X^3_{(2)}\times Y_3,
\]
satisfy the same commutation relations as $Y_1,Y_2,$ and $Y_3$. In consequence, $ X^E_t= M^3+b_1(t)M^1$,
span a generalized distribution $\mathcal{D}^{X^E}$ of rank three at a generic point of ${T}\mathbb{R}^2\times \mathcal{O}_2$. As this manifold is seven-dimensional and the differential of the $t$-independent first-integrals of $X$ must vanish on vector fields taking values on the integrable distribution $\mathcal{D}^{X^E}$,  we obtain that $X^E$ admits four (locally defined) $t$-independent functionally independent first integrals. Moreover, since ${\rm pr}_*:\mathcal{D}^{X^E}\rightarrow {T}\left({T}\mathbb{R}^2\right)$, with ${\rm pr}:(x_{(1)},v_{(1)},x_{(2)},v_{(2)},x,v,a)\in {T}\mathbb{R}^2\times\mathcal{O}_2\mapsto (x_{(1)},v_{(1)},x_{(2)},v_{(2)})\in {T}\mathbb{R}^2$, is injective at each point of an open dense subset of ${T}\mathbb{R}^2\times \mathcal{O}_2$, we can ensure that the system ${X^E}$ possesses a mixed superposition rule (cf. \cite{GL12}).

Standard techniques to obtain a mixed superposition rule for $X^E$ demand the integration of the vector fields $M^1$, $M^2,$ and $M^3$, e.g., by means of the
method of characteristics \cite{GL12}. We here propose a simpler method based on the fact that $X$, $X_{(1)},$ and $X_{(2)}$ are Dirac--Lie systems.
More specifically, $X$ is  a Dirac--Lie system with respect to
$
\omega=v^{-3}dv\wedge da
$
and $X_{(1)},$ $X_{(2)}$ with respect to
$\omega_{(1)}=dx_{(1)}\wedge dv_{(1)}$ and $\omega_{(2)}=dx_{(2)}\wedge dv_{(2)}$, respectively. Using this, we can define on ${T}\mathbb{R}\times{T}\mathbb{R}\times\mathcal{O}_2$
the closed two-form
\[
\omega_1=v^{-3}dv\wedge da+dx_{(1)}\wedge dv_{(1)}.
\]
Now, since $Y_1$, $Y_2,$ and $Y_3$ have $L^\omega$-Hamiltonian functions (\ref{Fun3KS}) and $X^1_{(i)},X^2_{(i)},X^3_{(i)}$ have
$L^{\omega_{(i)}}$-Hamiltonian functions
\[
h_1=-\frac{x_{(i)}^2}{2},\qquad h_2=\frac {1}{2}x_{(i)}v_{(i)},\qquad h_3=-\frac{v_{{(i)}}^2}{2},\qquad i=1,2\,,
\]
the vector fields $M^1$, $M^2,$ and $M^3$ admit the $L^{\omega_1}$-Hamiltonian functions
\[
h_1=-\frac 2v-\frac 12 x_{(1)}^2,\qquad h_2=-\frac{a}{v^2}+\frac 12 x_{(1)}v_{(1)},\qquad h_3=-\frac {a^2}{2v^3}-\frac{v_{(1)}^2}{2},
\]%
which close the same commutation relations (relative to the Poisson bracket induced by $\omega_1$) as $M^1$, $M^2,$ and $M^3$ with respect to the Lie bracket of vector fields. Thus, $h_1,h_2$ and $h_3$ span a finite-dimensional real Lie algebra of functions isomorphic to $\mathfrak{sl}(2,\mathbb{R})$. In consequence, the function
\[
F_1=h_1h_3-h_2^2=\frac{(2vv_{(1)}+ax_{(1)})^2}{4v^3}
\]
Poisson commutes with $h_1$, $h_2$, and $h_3$, so $F_1$ is a constant of the motion for $X^E$.
Proceeding in a similar way with the closed two-form
\[
\omega_2=v^{-3}dv\wedge da+dx_{(2)}\wedge dv_{(2)},
\]
we obtain a new constant of the motion
\[
F_2=\frac{(2vv_{(2)}+ax_{(2)})^2}{4v^3}
\]
for $X^E$. In order to obtain a mixed superposition rule, we need a third common $t$-independent constant of the motion for $M^1,M^2,$ and $M^3$. This can be done
by recalling that Schwarzian equations admit a Lie symmetry $Z_P$ given by (\ref{LieSym}) and the systems (\ref{linear}) have the Lie symmetry
\[
Z_L=\frac 12 \left(x_{(1)}\frac{\partial}{\partial x_{(1)}}+v_{(1)}\frac{\partial}{\partial v_{(1)}}+x_{(2)}\frac{\partial}{\partial x_{(2)}}+v_{(2)}\frac{\partial}{\partial v_{(2)}}\right).
\]
Hence, $Z_P\times Z_L$ is a Lie symmetry for $X^E$.
Using the method employed in the last section, we have that $({T}\mathbb{R}\times{T}\mathbb{R}\times \mathcal{O}_2,L^{\omega_{Z_P\times Z_L}},X_{(1)}\times X_{(2)}\times X)$ is a Dirac--Lie system with
\[
\omega_{Z_P\times Z_L}\equiv \mathcal{L}_{Z_P\times Z_L}\omega_{1}=\frac{2}{v^3}(xda\wedge dv +vdx\wedge da+adv\wedge dx)+dx_{(1)}\wedge dv_{(1)}.
\]
As Hamiltonian functions for $Z_1$, $Z_2,$ and $Z_3$ can be taken
\[
h_1=\frac {4x}{v}-\frac 12 x_{(1)}^2,\qquad h_2=-2+\frac{2ax}{v^2}+\frac 12 x_{(1)}v_{(1)},\qquad h_3=-\frac{2a}{v}+\frac {a^2x}{v^3}-\frac{v_{(1)}^2}{2}.
\]
These functions span a Lie algebra isomorphic to $\mathfrak{sl}(2,\mathbb{R})$. In consequence, we obtain through the corresponding Casimir the constant of the motion
\[
I=-\frac{(2vv_{(1)}\!+\!a x_{(1)})(2v_{(1)}vx-2v^2x_{(1)}+ax x_{(1)})}{2v^3}\!=\!2F_1\left(\!-x\!+\!\frac{2x_{(1)}v^2}{2vv_{(1)}+ax_{(1)}}\right).
\]
As $F_1$ is a constant of the motion, we obtain that
\[
F_3=-x+\frac{2x_{(1)}v^2}{2vv_{(1)}+ax_{(1)}}
\]
is a much simpler constant of the motion which will simplify further calculations. Note that
\[
\frac{\partial(F_1,F_2,F_3)}{\partial(x,v,a)}\neq 0.
\]
Hence, we can make use of $F_1,F_2,$ and $F_3$  to obtain a mixed superposition rule for $X$ in terms of $X_{(1)}$ and $X_{(2)}$. More specifically, by imposing $F_3=\lambda_3,$ with $\lambda_3$ being a certain real constant, we obtain
\[
x=-\lambda_3+\frac{2v^2x_{(1)}}{2vv_{(1)}+ax_{(1)}}.
\]
Now, imposing $F_1=\lambda_1$ and $F_2=\lambda_2$, we see that
\begin{equation}\label{eq1}
2vv_{(1)}+ax_{(1)}=\pm2v\sqrt{\lambda_1 v},\qquad  2vv_{(2)}+ax_{(2)}=\pm2v\sqrt{\lambda_2 v}.
\end{equation}
For simplicity, we restrict ourselves to the case when the signs are positive. Multiplying the first equality in (\ref{eq1}) by $x_{(2)}$, the second by $x_{(1)}$, subtracting and using that $v\neq 0$, we obtain
\[
\begin{gathered}
v_{(1)}x_{(2)}-v_{(2)}x_{(1)}= x_{(2)}\sqrt{\lambda_1v}-x_{(1)}\sqrt{\lambda_2v}.
\end{gathered}
\]
Multiplying the first equality in (\ref{eq1})  by $v_{(2)}$, the second by $v_{(1)}$, and subtracting, we get
\[
\begin{gathered}
a(v_{(2)}x_{(1)}-v_{(1)}x_{(2)})= 2 v (v_{(2)}\sqrt{\lambda_1v}-v_{(1)}\sqrt{\lambda_2v}).
\end{gathered}
\]

If we assume that $\mathcal{W}=v_{(2)}x_{(1)}-v_{(1)}x_{(2)}\neq 0$, i.e.,  $(x_{(1)},v_{(1)})$ and $(x_{(2)},v_{(2)})$ are not proportional, then
\[
v={\rm sg}(\lambda_1)\frac{(v_{(2)}x_{(1)}-v_{(1)}x_{(2)})^2}{[x_{(2)}\sqrt{|\lambda_1|}-x_{(1)}\sqrt{|\lambda_2|}]^2}.
\]
From this,
\begin{equation}\label{eq3}
\begin{gathered}
x=-\lambda_3+ {\rm sg}(\lambda_1)\left|\frac{v_{(2)}x_{(1)}-x_{(2)}v_{(1)}}{x_{(2)}\sqrt{|\lambda_1|}-
x_{(1)}\sqrt{|\lambda_2|}}\right|\frac{x_{(1)}}{\sqrt{|\lambda_1|}}, \\
\quad a=-2{\rm sg}(\lambda_1)\frac{(v_{(2)}x_{(1)}-x_{(2)}v_{(1)})^2}{(x_{(2)}\sqrt{|\lambda_1|}-
x_{(1)}\sqrt{|\lambda_2|})^3}({v_{(2)}\sqrt{|\lambda_1|}-v_{(1)}\sqrt{|\lambda_2|}}).
\end{gathered}
\end{equation}
Previous expressions give rise to a mixed superposition rule for system (\ref{firstKS3}) with $c_0=0$ in terms of two linearly independent particular solutions of the systems $X_{(1)}$ and $X_{(2)}$. More specifically, the mapping $\Phi:(x_{(1)},v_{(1)},x_{(2)},v_{(2)};\lambda_1,\lambda_2,\lambda_3)\in{T}\mathbb{R}^2\times \mathbb{R}^3\mapsto (x,v,a)\in \mathcal{O}_2$ allows us to bring the general solution of $X$ into the form
\[
(x(t),v(t),a(t))=\Phi(x_{(1)}(t),v_{(1)}(t),x_{(2)}(t),v_{(2)}(t);\lambda_1,\lambda_2,\lambda_3).
\]
Moreover, we can further simplify the form of $\Phi$. Observe that $x_{(1)}(t)v_{(2)}(t)-x_{(2)}(t)v_{(1)}(t)$, where $(x_{(i)}(t),v_{(i)}(t))$, $i=1,2$, are particular solutions of $X_{(1)}$ and $X_{(2)}$, is the
Wronskian associated to two particular solutions $x_{(1)}(t),x_{(2)}(t)$ of $d^2x/dt^2=-b_1(t)x$.

It is interesting to note that the map $\tau^{(2)}\circ\Phi$, where $\tau^{(2)}:(x,v,a)\in {T}^2\mathbb{R}\mapsto x\in\mathbb{R}$ is the projection of the
second tangent bundle ${T}^2\mathbb{R}$ onto $\mathbb{R}$, describes the general solution of Schwarzian equations in terms of particular solutions of other systems.
We could say that this is an example of a mixed superposition rule for higher-order systems of differential equations, which could be used to generalize the notion of superposition rules for higher-order systems of differential equations proposed in \cite{CGL11}.

\section{Dirac--Lie systems and Schwarzian--KdV equations}
Let us give some final relevant applications of our methods. In particular, we devise a procedure to construct traveling wave solutions for some relevant nonlinear PDEs by means of Dirac--Lie systems. For simplicity, we hereafter denote the partial derivatives of a function $f:(x_1,\ldots,x_n)\in \mathbb{R}^n\mapsto f(x_1,\ldots,x_n)\in\mathbb{R}$ in the form $\partial_{x_i}f$.

Consider the so-called {\it Schwarzian Korteweg de Vries equation} (SKdV equation)\cite{I10}
\begin{equation}\label{SKdV}
\{\Phi,x\}\partial_x\Phi=\partial_t\Phi,
\end{equation}
where $\Phi:(t,x)\in\mathbb{R}^2\rightarrow \Phi(t,x)\in\mathbb{R}$ and
\[
\{\Phi,x\}\equiv \frac{\partial^3_x \Phi}{\partial_x\Phi}-\frac 32\left(\frac{\partial^2_x\Phi}{\partial_x \Phi}\right)^2.
\]
This PDE has  been attracting some attention due to its many interesting properties \cite{I10,GRL98,RBMG03}. For instance, Dorfman established
a bi-symplectic structure for this equation \cite{Do89}, and many others have been studying its solutions and generalizations \cite{I10,RBMG03}. As a relevant result, we can mention that, given a solution $\Phi$ of the SKdV equation, the function $\{\Phi,x\}$ is a particular solution of the Korteweg-de Vries equation (KdV equation) \cite{Mar11}
\begin{equation}\label{KdV}
\partial_tu=\partial_x^3u+3u\partial_xu .
\end{equation}

We now  look for traveling wave solutions of (\ref{SKdV}) of the type $\Phi(t,x)=g(x-f(t))$ for a certain fixed $t$-dependent function $f$ with $df/dt=v_0\in \mathbb{R}$. Substituting $\Phi=g(x-f(t))$ within (\ref{SKdV}), we obtain that $g$ is a particular solution of the Schwarzian equation
\begin{equation}\label{Traveling}
\frac{d^3g}{dz^3}=\frac 32\frac{(d^2g/dz^2)^2}{dg/dz}-v_0\frac{dg}{dz},
\end{equation}
where $z\equiv x-f(t)$. We already know that the Schwarzian equations can be studied through the superposition rule (\ref{eq3}), which can better be obtained by using that Schwarz equations can be studied through a Dirac--Lie system, as seen in this work. More specifically, we can generate all their solutions from a known one as
\begin{equation}\label{Sup}
g_2(z)=\frac{\alpha g_1(z)+\beta}{\gamma g_1(z)+\delta},\qquad \alpha\delta-\beta\gamma\neq 0,\qquad \alpha,\beta,\gamma,\delta\in\mathbb{R}.
\end{equation}
In addition, (\ref{Traveling}) is a HODE Lie system, i.e., when written as a first-order system
by adding the variables $v=dx/dz$ and $a=dv/dz$, it becomes a Lie system $X$,
namely one of the form (\ref{Ex}). It can be proved that (\ref{Traveling}) can be integrated for any $v_0=df/dt$.
For instance, particular solutions of this equation read
\[
\bar{g}_1(z)={\rm th}\left[\sqrt{v_0/2}z\right]\,\,\,\, (v_0>0),\qquad g_1(z)=\frac{1}{z+1}\,\,\,\,(v_0=0).
\]
Note that $g_1(z)$ has the shape of a solitary stationary solution, i.e., $\lim_{x\rightarrow \pm\infty}g_1(x-\lambda_0)=0$  for every $\lambda_0\in\mathbb{R}$. Meanwhile, $\bar{g}_1$ is a traveling wave solution. Moreover, the general solution of (\ref{Traveling}) in both cases can be obtained from (\ref{Sup}).

Other methods can be employed to study SKdV equations through the Lie system (\ref{Traveling}).
For instance, our mixed superposition rule allows us to obtain the general solution of (\ref{Traveling}) out of a couple of particular solutions of the linear system (\ref{linear}). Obviously, this can be much easier than solving (\ref{Traveling}) directly.

Finally, it is also relevant that every Lie system related to a Lie algebra of vector fields $V$ induces the so-called {\it quasi-Lie scheme} $S(V,V)$ \cite{CGL11}.  One of the reasons of the importance of this scheme is that it induces a group $\mathcal{G}(V)$ of $t$-dependent changes of variables, the referred to as the {\it group of the scheme}, that enables us to transform the system $X$ into a new Lie system with the same Vessiot--Guldberg Lie algebra. This can be potentially employed to transform $X$ into a new Schwarzian equation with a different $f(t)$, which would be rise to a certain type of B\"acklund transformations for our traveling solutions of SKdV equations.

\section{Conclusions and Outlook}

We have introduced a new type of Lie systems on Dirac manifolds generalizing
Lie--Hamilton systems. We have analyzed their geometric properties and we showed that they can be employed to study systems, e.g. SKdV and Schwarzian equations, appearing in the
physics and mathematics literature.

In addition, the more general structure of Dirac--Lie systems allowed us to
investigate systems that could not be treated through known techniques. In particular,
we have developed a theory to obtain bi--Dirac--Lie systems, several methods to obtain constants of the motion, Lie symmetries, and superposition rules for Dirac--Lie systems, and various generalizations of notions appearing in the theory of Lie systems. As a result, we were
able to obtain through geometric and algebraic techniques  mixed and standard superposition rules for Schwarzian equations.

In a future research, we will aim at finding new types of Lie systems related to other geometric
structures. For instance, it
would be interesting to study the existence of Lie systems admitting a
Vessiot--Guldberg Lie algebra of Hamiltonian vector fields
with respect to almost or twisted Poisson structures. We are also interested in developing a generalization of the theory of this work to the framework of Lie algebroids.
The latter has shown to be very fruitful in Geometric Mechanics \cite{LMM05,GG08,GGU06,GLMM09} and Control Theory \cite{CM04,GJ11}.
Moreover, a further analysis of the properties of Dirac--Lie
systems is being performed. Moreover, we aim to develop co-algebra techniques \cite{BBR06,BCHLS13} to
obtain mixed and standard superposition rules for Dirac--Lie systems.

\section*{Acknowledgements}
Research of J. Grabowski and J. de Lucas founded by the  Polish National Science Centre grant
HARMONIA under the contract number DEC-2012/04/M/ST1/00523.
Research of J.F. Cari\~nena and J. de Lucas was partially financed by research
projects MTM2012-33575 and MTM2011-15725-E (MEC) and E24/1 (Gobierno de Arag\'on). C. Sard\'on acknowledges a fellowship provided by the University of Salamanca and partial financial support by
research project FIS2009-07880 (DGICYT).

\bibliographystyle{elsarticle-num}

\end{document}